\documentclass[aps,pra,twocolumn,superscriptaddress,floatfix,
nofootinbib,showpacs,longbibliography]{revtex4-1}
\usepackage[utf8]{inputenc}  
\usepackage[T1]{fontenc}     
\usepackage[british]{babel}  
\usepackage[sc,osf]{mathpazo}
\usepackage[scaled=0.86]{berasans}  
\usepackage[colorlinks=true, citecolor=blue, urlcolor=blue]{hyperref}  
\usepackage{graphicx} 
\usepackage[babel]{microtype}  
\usepackage{amsmath,amssymb,amsthm,bm,amsfonts,mathrsfs,bbm} 

\usepackage{xspace}  
\usepackage{pgf,tikz}
\usepackage{xcolor}
\usepackage{multirow}
\usepackage{array}
\usepackage{bigstrut}
\usepackage{braket}
\usepackage{color}
\usepackage{natbib}
\usepackage{multirow}
\usepackage{mathtools}
\usepackage{float}
\usepackage{soul}
\usepackage{xcolor,colortbl}
\usepackage{color}
\usepackage[justification=justified, format=plain]{subcaption}
\usepackage[justification=raggedright]{caption}

\newcommand{\be}{\begin{equation}}
\newcommand{\ee}{\end{equation}}
\newcommand{\ba}{\begin{eqnarray}}
\newcommand{\ea}{\end{eqnarray}}

\newtheorem{theorem}{Theorem}
\newtheorem{corollary}{Corollary}
\newtheorem{definition}{Definition}

\newtheorem{remark}{Remark}

\allowdisplaybreaks

\def\>{\rangle}
\def\<{\langle}

\newtheorem{prop}{Proposition}

\begin{document}
	
\title{Quantum superposition of causal structures as a universal resource for local implementation of nonlocal quantum operations}	

\author{Pratik Ghosal}
\email{ghoshal.pratik$00$@gmail.com}    
\affiliation{Department of Physics, Bose Institute, Unified Academic Campus, Block EN, Sector V, Bidhan Nagar, Kolkata 700 091, India}

\author{Arkaprabha Ghosal}
\email{a.ghosal$1993$@gmail.com}
\affiliation{Department of Physics, Bose Institute, Unified Academic Campus, Block EN, Sector V, Bidhan Nagar, Kolkata 700 091, India}
\affiliation{Optics \& Quantum Information Group, The Institute of Mathematical Sciences, CIT Campus, Taramani, Chennai 600113, India}

\author{Debarshi Das}
\email{dasdebarshi$90$@gmail.com}
\affiliation{Department of Physics and Astronomy, University College London, Gower Street, WC1E 6BT London, England, United Kingdom}

\author{Ananda G. Maity}
\email{anandamaity$289$@gmail.com}
\affiliation{Networked Quantum Devices Unit, Okinawa Institute of Science and Technology Graduate University, Onna-son, Okinawa 904-0495, Japan}

\begin{abstract}
Spatial separation restricts the set of locally implementable quantum operations on distributed multipartite quantum systems. We propose that indefinite causal structure arising due to quantum superposition of different space-time geometries can be used as an independent universal resource for local implementation of any quantum operation on spatially distributed quantum systems. Consequently, all such quantum tasks that are not accomplishable by local operations and classical communication (LOCC) only also become locally accomplishable. We show that exploiting indefinite causal structure as the sole resource, it is possible to perfectly teleport the state of one agent's subsystem to the other distant laboratory in such a way that the agent at the distant laboratory can have access to the whole initially shared state in his or her laboratory and can perform any global quantum operation on the joint state locally. We further find that, after the teleportation process, the resource -- indefinite causal structure of the space-time does not get consumed. Hence, after implementing the desired quantum operation the state of the first agent's subsystem can be teleported back to its previous laboratory using the same resource. We show that this two-way teleportation is not always necessary for locally executing all nonlocal quantum tasks that are not realisable by LOCC only. Without invoking any kind of teleportation, we present a protocol for perfect local discrimination of the set of four Bell states that exploits indefinite causal structure as the sole resource. As immediate upshots, we present some more examples of such nonlocal tasks as local discrimination of the set of states exhibiting ``quantum nonlocality without entanglement" and activation of bound entangled states that are also achievable by our proposed protocol incorporating indefinite causal structure as a resource.
\end{abstract}
\maketitle

\section{Introduction}
Fundamental indefiniteness in the causal structure of space-time is a unique phenomenon that arises in quantum theory of gravity \cite{Hardy05,Hardy09,hardy2007towards,Brukner18}. In general relativity, the causal structure of space-time is definite and is determined by the distribution of matter-energy, which, although is assumed to be classical, can be dynamical in general. However, when the degrees of freedom of the matter-energy are taken to be quantum, this classical description of space-time is no longer tenable. Due to the presence of quantum matter, space-time with different geometries may exist in superposition, leading to the causal structure to be both dynamical (as in general relativity) as well as indefinite (as in quantum theory) \cite{Hardy05,Brukner18,Castro-Ruiz18}. Such a revolutionary concept, where the interplay between general relativity and quantum theory provides a fundamentally indefinite causal structure, was first proposed by Hardy \cite{Hardy05}. One striking consequence of this indefiniteness of the causal structure is that the causal order between two events in the space-time can also become indefinite.

A general question that naturally arises is whether this aforementioned indefiniteness can be utilised as useful resource for information processing tasks \cite{Hardy09}. Such application of indefinite causal order between two events was first manifested through the concept of quantum SWITCH \cite{Chiribella13}. In a quantum SWITCH, an auxiliary system is used to control the order of two quantum operations (each being associated with an event) acting on a single quantum system. Consequently, the auxiliary system prepared in a superposition of two orthogonal states, each corresponding to one of the two different orders of the two operations, makes their order indefinite. This type of coherent control of order of events has also been introduced in the process matrix formalism \cite{Oreshkov12}, capturing the most general way in which local operations can be implemented in spatially separated laboratories without assuming any predefined causal order between them. Indefiniteness of causal order of two events have been shown to be beneficial for several information processing tasks ranging from testing the properties of quantum channel \cite{Chiribella12}, winning of noncausal games \cite{Oreshkov12}, minimising quantum communication complexity \cite{Guerin16}, boosting the precision of quantum metrology \cite{Zhao20}, improving quantum communication \cite{Ebler18,Chiribella21,Mukhopadhyay19,koudia2021causal,Banik21,Das21}, to achieving quantum computational \cite{Araujo14} and thermodynamic advantages \cite{Tamal20,Vedral20,Maity22}. Experimental realisations of such advantages have also been reported \cite{Procopio15,Rubino17,Goswami18(1)}.

Another interesting application of indefinite causal structure is manifested in assessing the quantum nature of gravitational interaction. Very recently, two proposals have been put forward for testing the quantum nature of gravity in table-top experiments, based on quantum information theoretic arguments \cite{Bose17,Marletto17}. In these proposals, creation of entanglement between two initially separable distant masses via their mutual gravitational interaction is claimed to witness the quantum signature of gravitational interaction. It has been shown that behind these proposals, quantum superposition of space-time geometries plays an important role in the entanglement generation \cite{Christodoulou:2018cmk}.

Motivated by such wide applications of  indefinite causal structure of space-time,  in this article we address its role as a resource for local implementation of nonlocal quantum operations on quantum systems distributed among multiple spatially separated parties. Here, by the term ``nonlocal quantum operations", we  mean only those quantum operations that are not implementable by Local Operation and Classical Communication (LOCC) alone on the spatially separated quantum systems.

Since the set of quantum operations accomplishable on spatially separated quantum systems by LOCC is a strict subset of all quantum operations, only some limited quantum tasks can be performed locally by spatially separated parties \cite{Chitambar2014} in the absence of any quantum correlation (or, more generally, any quantum resource). For example, if the parties share a quantum system in a product state, they will never be able to perform quantum teleportation of an unknown quantum state using LOCC only. However, with the assistance of suitable resources (for example, entanglement), these limitations can be lifted. We term the class of quantum tasks that cannot be executed by spatially separated parties using LOCC, in absence of any additional quantum resource, as ``nonlocal quantum tasks". It is important to mention here that there are several tasks (e.g., local state discrimination) that are impossible by LOCC alone, if only a single copy of the shared state is given. However, when multiple copies are accessible, then those tasks become possible without using any additional resource. For such quantum tasks, in this article, we explicitly make the assumption that only a single copy of the shared state is available which necessitates the use of an additional resource. Furthermore, note that there are several quantum tasks which may be accomplished in the absence of any quantum resource when classical communication is allowed (e.g., the  Clauser-Horne-Shimony-Holt game or CHSH game). Such tasks are also popularly known as nonlocal quantum tasks. However, by definition, we exclude those tasks from the set of our defined nonlocal quantum tasks.

Now, instead of focusing on any specific nonlocal task, we address whether any generic nonlocal quantum task (as defined earlier) can be accomplished by spatially separated classically communicating agents acting on shared quantum system when they have access to no resource except the indefiniteness of the causal structure of space-time. It turns out quite naturally that if the two agents can locally implement any desirable quantum operation on the state of their shared system, they would be able to successfully accomplish any nonlocal quantum task on it.

Let us now briefly explain how indefinite causal structure acts as a resource in the aforementioned contexts. We consider bipartite quantum systems shared between two spatially separated agents in a space-time with indefinite causal structure. We devise a protocol that the two agents, limited only to LOCC, follow to exploit this indefiniteness of causal-structure to perform any nonlocal operation as well as task. We show that using our proposed protocol, it is possible to perfectly teleport (with unit teleportation fidelity) any quantum state to distant laboratory without the necessity of any other resource such as entanglement. In particular, starting from a bipartite system $(S_1,S_2)$ shared between two spatially separated agents, say, Alice and Bob and an ancillary system $S_3$ belonging to Bob, Alice can teleport the state of $S_1$ to Bob's laboratory exploiting indefinite causal structure, such that the joint state of $(S_1,S_2)$ now becomes the state of $(S_3,S_2)$. Bob then performs the desired global quantum operation on the joint state of $(S_3,S_2)$. This teleportation protocol using indefinite causal structure as the sole resource is fairly novel and important in its own right. Interestingly, the indefiniteness of the causal structure of space-time is regenerated at the end of this teleportation process, which implies that it acts like a catalyst. Hence, after the implementation of the desired quantum operation, Bob can teleport back the state of $S_3$ to Alice, such that the transformed joint state of $(S_3,S_2)$ after the quantum operation becomes the shared state of $(S_1,S_2)$. Thus, indefinite causal structure of the space-time serves as a universal resource for local implementation of any quantum operation on the shared quantum systems, thereby making local execution of any nonlocal quantum task possible. 

It is noteworthy here that it is possible to generate a maximally entangled state between two distant locations using indefinite causal structure as resource and use that shared maximally entangled state for standard teleportation protocol \cite{Bennett1993} and execution of other nonlocal quantum tasks. In that case, the entanglement plays the role of an indispensable resource and the role of indefinite causal structure is only to generate the entanglement. This leads to the ambiguity about which resource is the primary one. However, the teleportation protocol that we devise fundamentally differs from the standard teleportation protocol \cite{Bennett1993} as it bypasses the necessary requirement of shared maximally entangled state and uses indefinite causal structure as the sole resource. In fact, shared maximally entangled state never arises in our teleportation protocol. This clearly demonstrates that indefinite causal structure can also be treated as a fundamental resource which can be used independent of entanglement for local implementation of nonlocal quantum operations and tasks. Furthermore, by presenting an example of nonlocal task viz., perfect local discrimination of the set of four Bell states, we reaffirm the independence of indefinite causal structure as a resource, and that the aforementioned back and forth teleportation is not always a necessary step for executing nonlocal tasks using LOCC and indefinite causal structure. As offshoots, we further show that indefinite causal structure can be utilised for execution of some other nonlocal tasks involving more than two parties, like perfect local discrimination of the set of tripartite orthogonal product states exhibiting quantum nonlocality without entanglement \cite{Bennet99}, and local activation of a four-qubit bound entangled state (Smolin state) \cite{Smolin2001}, when only two of the spatially separated parties have access to indefinite causal structure as resource.

The outline of this paper is as follows. In Sec. \ref{two}, we discuss nonlocal quantum operations and the role of resources for local implementation of those nonlocal operations. We then also discuss in fair details how quantum matter-energy degrees of freedom and, in particular, spatial superposition of gravitating mass can lead to the indefiniteness of causal structure of space-time. In Sec. \ref{three}, we design our protocol that exploits the indefiniteness of causal structure. In Sec. \ref{four}, we propose a perfect teleportation protocol using indefinite causal structure as resource that allows spatially separated parties to locally implement any nonlocal operation on shared quantum systems. This establishes indefinite causal structure as  a universal resource for local implementation of any nonlocal quantum operation as well as task. We also illustrate here the independence of indefinite causal structure as a resource from entanglement. Finally, we conclude with discussion in Sec. \ref{five}.

\section{\label{two} Preliminaries}
Before going into the details of our study, we present some preliminary concepts that will be useful in the main analyses.

\subsection{\label{nonlocal} Nonlocal quantum operations and resources}
Any possible transformation of an arbitrary quantum state is described by the action of quantum operations on it. If $D(\mathcal{H})$ denotes the set of density operators on a Hilbert space $\mathcal{H}$, then a quantum operation $\Phi$ is defined as a completely positive (CP) trace non-increasing map from $D(\mathcal{H})$ to $D(\mathcal{H})$ i.e., $\Phi : D(\mathcal{H})\mapsto  D(\mathcal{H})$. Let us denote the set consisting of all such quantum operations by $\mathcal{O}$.  When a multipartite quantum system in an arbitrary quantum state is distributed among multiple spatially separated parties, then the set of locally implementable quantum operations on the composite system by the spatially separated parties are limited due to the constraint of spatial separation. If each party locally implements quantum operations on their respective subsystem and communicates classical information between each other, then the set of operations implementable on the state of the composite system, called LOCC (denoted by $\mathcal{L}$) forms a strict subset of the whole set of quantum operations, i.e.,  $\mathcal{L}\subset \mathcal{O}$. Quantum operations that are not implementable by LOCC only are defined as nonlocal operations ($\mathcal{L}^c$). This inability to locally implement the full set of quantum operations restricts the spatially separated parties from accomplishing many quantum tasks.  To illustrate, consider the local state discrimination task as described below. Given a single copy of a set of orthogonal states, it is always possible to perfectly distinguish the states if the complete set of quantum operations are allowed to be  implemented on them. However, if only a restricted set of operations is allowed to be implemented on the states, then they cannot necessarily be distinguished perfectly. In particular, there exists certain sets of multipartite orthogonal states that are not distinguishable perfectly when spatially separated parties are limited to LOCC only \cite{PhysRevLett.87.277902,Bennet99}. Another such example of task which is impossible to execute locally is the transformation of product states shared between spatially separated parties to entangled states by LOCC \cite{RevModPhys.81.865}. As mentioned earlier, such tasks that cannot be accomplished by spatially separated parties under LOCC only are termed as nonlocal tasks.

Importantly, with the assistance of suitable additional resources these limitations for local implementation of nonlocal operations and execution of nonlocal tasks can be lifted \cite{PhysRevA.64.052309,PhysRevA.64.032302,PhysRevA.62.052317}. Consider, for example, two spatially separated parties Alice and Bob share a bipartite quantum state on which an arbitrary quantum operation has to be implemented. If Alice and Bob have sufficient number of entangled states shared between them in addition to the shared bipartite system, then they can locally implement any quantum operation by the following simple protocol. Alice will first teleport the state of her subsystem to Bob's laboratory so that Bob can have access to the full initially shared state, in his local laboratory. He will then implement any quantum operation from the set $\mathcal{O}$ on the composite system in his laboratory. Finally, after performing quantum operation, Bob will teleport the state of Alice's subsystem back to her laboratory. Hence, by teleporting back and forth, Alice and Bob will be able to implement the whole set of quantum operations locally \cite{PhysRevA.64.032302}. Here, the shared entangled states serve as resource for teleportation and, hence, for local implementation of the complete set of quantum operations. In case of multipartite (say, $n$ parties) systems, a similar protocol can be followed where all the $n-1$ parties will teleport the states of their respective subsystems to one party, say, the $n^{\text{th}}$ party, so that the $n^{\text{th}}$ party can implement any desirable quantum operation on the composite system and then teleport back the states of all the $(n-1)$ subsystems to the respective $(n-1)$ number of parties. In this paper, however, our main focus will be on bipartite systems. 

It may be noted here that for execution of all nonlocal tasks, the aforementioned  two-way teleportation is not always necessary. For example, in case of local state discrimination task, if Alice and Bob are able to perform quantum teleportation only in one direction (either from Alice to Bob, or from Bob to Alice), then they can distinguish the states from any set of bipartite orthogonal states. Here, after one-way teleportation, both the subsystems will be in a single laboratory so that complete measurement required for the state discrimination can be implemented on the joint state locally. However, there also exist tasks, like implementation of SWAP operation ($|\psi\rangle_A \otimes |\phi\rangle_B \mapsto |\phi\rangle_A \otimes |\psi\rangle_B$) where two-way quantum teleportation is necessary and sufficient \cite{PhysRevA.64.032302}. More generally, the specificity of the required resource depends on the task concerned  and the quantum state of the shared system on which the task is to be performed. This motivates us to define the notion of {\it universal} resource for local implementation of nonlocal quantum operations on shared bipartite systems.

\begin{definition}[Universal Resource] 
Consider Alice and Bob share a bipartite quantum state $\rho_{AB} \in D(\mathbb{C}_A^{d}\otimes\mathbb{C}_B^{d'})$. We call a resource \textbf{universal} for the said Hilbert space if Alice and Bob can implement the whole set of quantum operations on any such shared quantum state by LOCC with the aid of that resource only.
\end{definition}

For example, in case of $d \times d'$ dimensional bipartite quantum systems with $d\leq d'$, two maximally entangled states in $ \mathbb{C}^{d}\otimes\mathbb{C}^{d}$ serve as a universal resource. It trivially follows that if Alice and Bob can locally implement the complete set of quantum operations with the aid of a universal resource, then they will also be able to execute any nonlocal task. Therefore, the resource is also universal for local execution of any nonlocal task. 

\subsection{\label{sup} Superposition of space-time with different geometries}
According to general relativity, the geometry of space-time is uniquely determined by the distribution of matter-energy degrees of freedom in it as suggested by Einstein's field equations. This geometry regulates the configuration of light cones of every event in the space-time, i.e., the causal structure of the entire space-time. The light cone of an event $X$ specifies the set of events that are causally connected to it and also their causal orders with respect to it. Events that are within the light cone of $X$ are called timelike separated from it. Among these events, those which are in the forward light cone are said to be in the future of $X$ and those which are in the backward light cone are said to be in the past of $X$. Note that it is always possible to signal (for example, sending a light pulse) from an event in past to an event in future, but not the other way. For two timelike separated events $X_1$ and $X_2$, we denote their causal order as $X_i \rightarrow X_j$, $i,j \in \lbrace1,2\rbrace$, where $X_i$ is the event in past and $X_j$ is the event in future. In the classical description of gravity, causal order between any pair of events is definite and determined essentially by the classical matter-energy degrees of freedom. Thus, the causal structure of the space-time is albeit dynamic, classically deterministic in nature.

However, if the matter-energy degrees of freedom admits a quantum description, then the classical description of space-time may fall short. In particular, it can be anticipated that in such a scenario, geometry of the space-time may no longer be definite. For instance, if a gravitating mass exists in a quantum superposition of more than one distinct location, then space-time also exists in a quantum superposition of multiple different geometries. Such superposition of space-time geometries has got much attention in recent times, as evinced from a wide range of literature starting from quantum reference frame \cite{Giacomini19,Giacomini20} to exploring quantum signature of gravity \cite{Christodoulou:2018cmk, Bose17, Marletto17,Ruiz2017}. In fact, such superposition of the geometry of space-time can also lead to  entanglement between two distant gravitating masses, which is otherwise impossible in a classical space-time \cite{Christodoulou:2018cmk, Bose17, Marletto17}. Similar phenomena happen if gravitating matter-energy degree of freedom exists in a superposition of different energy eigenstates \cite{Ruiz2017}. As an immediate consequence of this superposition of the space-time geometry, one can have indefiniteness of causal structure of the space-time. Due to the superposition of geometries, the light cone of an event $X$ can also exist in a superposition of different configurations. Moreover, as mentioned in the introduction, an interesting special case of this type of indefiniteness is that the causal order between two events can also become indefinite. To put it specifically, in the classical scenario, depending on the geometry of space-time, two timelike separated events $X$ and $Y$ can have two distinct classically exclusive causal orders -- either (i) event $X$ is in the causal past of event $Y$ ($X \rightarrow Y$), or (ii) event $Y$ is in the causal past of $X$ ($Y \rightarrow X$). But if the two different geometries, each yielding one of the above two distinct causal orders between the events, exist in a superposition, then the causal orders also exist in a coherent superposition leading to indefinite causal order of the events.

For a tangible illustration, consider the following scenario \cite{zych2019bell}: two initially synchronised clocks (say, $A$ and $B$) are located in the gravitational field of a spherical mass $M$ along the same radial direction from the mass. Let us denote the radial distances of clock $A$ and clock $B$ from the mass by $r_A$ and $r_B$, respectively. Clock $A$ is situated closer to the mass and clock $B$ is located at a distance $h$ apart from clock $A$, i.e., $r_B=r_A+h$. As the clock $A$ is closer to the mass, due to gravitational time dilation, it ticks at a slower rate than the clock $B$. Since the two clocks were synchronised initially (at $t_A=t_B=0$), time read by clock $A$ ($t_A$) always lags behind the time read by clock $B$ ($t_B$). Let us define event $X$ at the location of clock $A$ when it reads the time $t_A=\tau^*$, and similarly event $Y$ at the location of clock $B$ when it reads $t_B=\tau^*$. For these two events to be causally connected, a light signal emitted from the location of clock $B$ at $t_B=\tau^*$ must reach the location of clock $A$ before $t_A=\tau^*$. 

The static and isotropic geometry of the space-time due to the spherical mass is described by the Schwarzschild metric:
\begin{align}
    ds^2 =& g_{tt}(r) c^2dt^2 + g_{rr}(r) dr^2 + g_{\theta \theta}(r) d\theta^2 +g_{\phi \phi}(r) d\phi^2 \nonumber\\
    =& -\left(1-\frac{2V(r)}{c^2}\right) c^2dt^2 + \frac{1}{\left(1-\frac{2V(r)}{c^2}\right)} dr^2 \nonumber \\
     &+ r^2 d\theta^2 +r^2\sin^2\theta d\phi^2,
\end{align}
where $V(r)=\frac{GM}{r}$ is the Newtonian gravitational potential due to the mass $M$ at a radial distance $r$ from it. Here, $\lbrace ct,r,\theta,\phi \rbrace$ is the spherical polar coordinate system used by an observer located in a faraway location, effectively free from the effects of the gravitational field. We also assume that both the clocks are located well outside the Schwarzschild radius ($R_s=\frac{2GM}{c^2}$) of the space-time, i.e., $r_A,r_B>R_s$. The coordinate time (measured by the far away observer) taken for the light signal to travel from $r_B$ to $r_A$ is
\begin{align}
T_c=\frac{1}{c} \int_{r_B}^{r_A} dr' \sqrt{-\frac{g_{rr}(r')}{g_{tt}(r')}}=\frac{1}{c} \int_{r_A+h}^{r_A} \frac{dr'}{\left(1-\frac{2V(r')}{c^2}\right)}.
\end{align}
The infinitesimal proper time interval at a distance $r$ from mass $M$ is $d\tau(r)=\sqrt{-g_{tt}(r)}dt$. Hence, the rates at which the clocks $A$ and $B$ tick are related as
\begin{align}
    dt_A=\sqrt{\frac{g_{tt}(r_A)}{g_{tt}(r_B)}}dt_B,
\end{align}
where $dt_{A/B}=d\tau(r_{A/B})$. Assuming that the initial synchronisation of the clocks matches with the coordinate time $t = 0$, the time at clock $A$ when the light signal reaches there is 
\begin{align}
t_A=\sqrt{-g_{tt}(r_A)}\left( \frac{\tau^*}{\sqrt{-g_{tt}(r_B)}}+T_c \right).
\end{align}
For the events $X$ and $Y$ to be causally connected, this time should precede the time $t_A=\tau^*$, i.e., 
\begin{align}
\label{tau*}
     \sqrt{-g_{tt}(r_A)}\left( \frac{\tau^*}{\sqrt{-g_{tt}(r_B)}}+T_c \right) \leq \tau^*, \nonumber\\
    \text{or,}~~~~~~\tau^*\geq \frac{\sqrt{-g_{tt}(r_A)}T_c}{1-\sqrt{\frac{g_{tt}(r_A)}{g_{tt}(r_A+h)}}}.
\end{align}
So, if we define the events $X$ and $Y$ such that the condition (\ref{tau*}) is satisfied, then the two events are causally connected and event $Y$ is in past of event $X$, i.e., $Y\rightarrow X$.

On the other hand, if we interchange the location of the two clocks, then clock $B$ will be closer to the mass $M$ and $r_A=r_B+h$. Now, due to gravitational time dilation clock $B$ will tick at a slower rate than clock $A$, i.e., time read by clock $B$ will always lag behind the time read by clock $A$. If we again define events $X$ and $Y$ at the locations of clocks $A$ and $B$ respectively, when they locally read the respective time $\tau^*$, then the events will be causally connected if $\tau^*$ satisfies Eq. (\ref{tau*}) with $r_A$ and $r_B$ being interchanged. In this case, event $Y$ will end up being in the future of event $X$, i.e., $X\rightarrow Y$. These are depicted in Fig. \ref{ICO}.

\begin{figure}[t!]
\includegraphics[height=5.7cm,width=9.2cm]{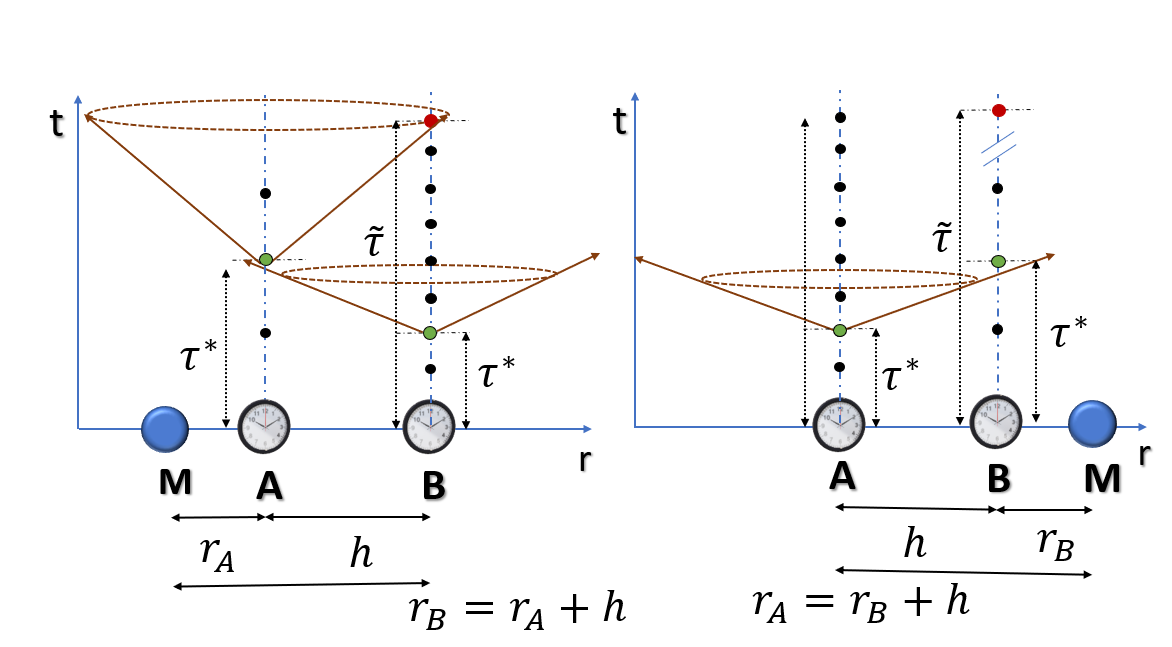}
\caption{\label{ICO} Two events $X$ and $Y$ (represented by green dots), defined at the locations of clock $A$ and $B$ as $t_A=\tau^*$ and $t_B=\tau^*$ respectively, have different causal orders depending on the position of gravitating mass $M$. The left picture is for the causal order where $Y$ is in the causal past of $X$ ($Y\rightarrow X$) and the right picture is for the causal order where $X$ is in the causal past of $Y$ ($X\rightarrow Y$). Event $Y'$ defined at the location of clock $B$ as $t_B=\tilde{\tau}$ (represented by red dot) is always in the causal future of $X$ ($X\rightarrow Y'$) in both the cases.}
\end{figure}

Now, let us consider the following scenario. Two clocks $A$ and $B$ are situated at a separation of distance $h$. A spherical mass $M$ exists in a spatial superposition of two locations -- one near the clock $A$, at a distance $R$ apart from it on the opposite side of clock $B$ and the other near clock $B$, at a distance $R$ apart from it on the opposite side of clock $A$ (see Fig. \ref{ICO1}). Note that here we are considering that the spherical mass is a quantum object, capable of existing in superposition of two different locations. As noted earlier, when the mass is located near the clock $A$, $r_A=R$, $r_B=R+h$ and events $X$ and $Y$ (defined as earlier) have the causal relation: $Y\rightarrow X$. We denote the quantum state of the mass for this location as $\ket{\mathcal{M}_{Y\rightarrow X}}$. On the other hand, when the mass is located near the clock $B$, $r_B=R$, $r_A=R+h$ and events $X$ and $Y$ have the causal relation: $X\rightarrow Y$. We denote the quantum state of the mass for this location as $\ket{\mathcal{M}_{X\rightarrow Y}}$. As the mass exists in a superposition of these two locations, its quantum state is given by $(\ket{\mathcal{M}_{X\rightarrow Y}}+\ket{\mathcal{M}_{Y\rightarrow X}})/\sqrt{2}$ and the causal order between $X$ and $Y$ becomes indefinite. Note that the states $\lbrace \ket{\mathcal{M}_{X\rightarrow Y}},\ket{\mathcal{M}_{Y\rightarrow X}} \rbrace$ yield observably distinct causal structures of space-time and are therefore orthogonal. Thus, fundamental indefiniteness of the causal order of events is a unique and natural feature of space-time in superposition of different geometries. It is noteworthy here that the quantum state $(\ket{\mathcal{M}_{X\rightarrow Y}}+\ket{\mathcal{M}_{Y\rightarrow X}})/\sqrt{2}$ of the mass (in superposition of two distinct location) yield indefinite causal order of not only two specific events, but of any pair of events defined as earlier, satisfying Eq. (\ref{tau*}) with $r_A=R$ and $r_B=R+h$.

However, the type of indefinite causal structure arising from the superposition of two geometries does not yield indefiniteness of causal order for more than two events. For example, if we consider three events -- $X$, $Y_1$, and $Y_2$ such that event $X$ is defined by some local time read by clock $A$ at its location and events $Y_1$ and $Y_2$ are defined by two different local times read by clock $B$ at its location, then we can get three distinct classically exclusive causal relations between them: $X\rightarrow Y_1\rightarrow Y_2$, $Y_1\rightarrow X\rightarrow Y_2$, and $Y_1\rightarrow Y_2\rightarrow X$ (causal order of $Y_1$ and $Y_2$ is always $Y_1\rightarrow Y_2$, as they are defined in the same location). These three causal relations arise from three different geometries of the space-time and, hence, indefinite causal order between the three events requires superposition of those three space-time geometries. So it is clear that the type of indefinite causal structure of space-time that allows indefiniteness of causal order between three events is different from the type of indefinite causal structure that allows indefiniteness of causal order between two events. Let us formalise this in the form of a definition:
\begin{definition}[$m$-ICS] Consider $n$ events $\lbrace X_k\rbrace_{k=1}^n$ in space-time. Now, depending on the space-time geometry, there can be at the most $n!$ distinct classically exclusive causal relations among the $n$ events. We denote the indefinite causal structure that allows $m$ (where $m \leq n!$) of these causal relations to exist in coherent superposition as $m$-Indefinite Causal Structure ($m$-ICS).
\end{definition}
Note that $m$-ICS arises from the superposition of $m$ distinct geometries, each yielding one of the $m$ causal orders for the $n$ events.

Another important point to mention here is that it is always possible to find a pair of events whose causal order is definite even in space-time with indefinite causal structure (discussed in details in Appendix \ref{appa}). This means that even in indefinite causal structure, it is possible to communicate definitely from a past event to a future event. This fact will be useful in subsequent sections of this paper.

\begin{figure}[t!]
\centering
\includegraphics[height=4cm,width=8.5cm]{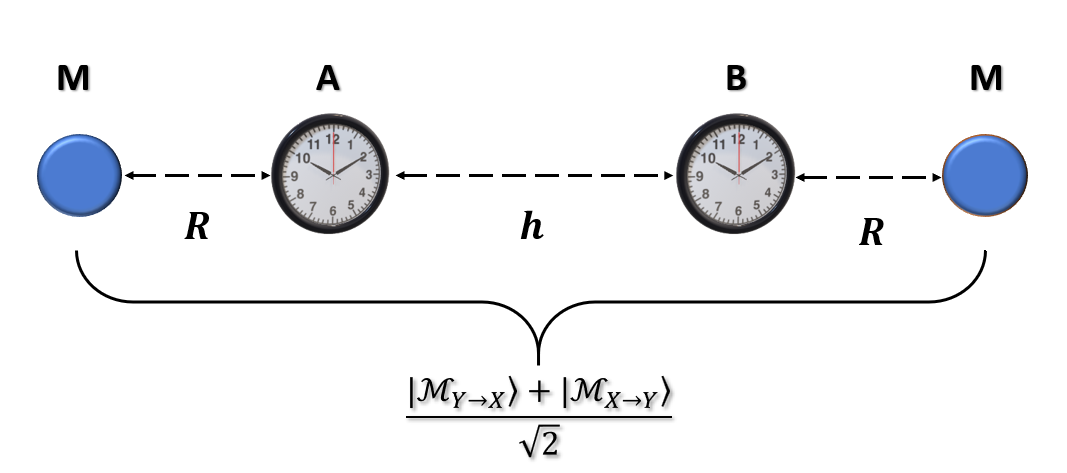}
\caption{\label{ICO1} Gravitating mass $M$ is in superposition of two distinct locations with states being denoted by $\ket{\mathcal{M}_{X\rightarrow Y}}$ and $\ket{\mathcal{M}_{Y\rightarrow X}}$. Two clocks $A$ and $B$ are situated in the region between the mass-locations, such that the distance between clock $A(B)$ and the mass-location with state $\ket{\mathcal{M}_{Y\rightarrow X}}(\ket{\mathcal{M}_{X\rightarrow Y}})$ is $R$ and that between the clocks $A$ and $B$ is $h$. Superposition of the two locations results in indefinite causal structure of space-time.}
\end{figure}

\section{\label{three} Protocol for exploiting indefinite causal structure as resource}
When two events are considered to be in classical space-time (with definite geometry), then one-way signaling (from the past event to the future event) is the only option between those two events. Instead, if two events are considered in space-time with indefinite causal structure, then two one-way signaling between the two timelike separated events exist in coherent superposition. This coherent superposition of two one-way signaling is quite different from the bidirectional signaling between two distant locations, where signaling from one location (say, $A$) to the other (say, $B$) takes place between two events $X$ and $Y$ defined at the locations of $A$ and $B$ respectively (such that $X\rightarrow Y$) and the back signaling from $B$ to $A$ takes place between a different set of events $Y$ and $X'$ defined at the locations of $B$ and $A$ respectively (such that $Y\rightarrow X'$) with $X'$ being an event which occurs at the same location of $X$ but after some sufficient time so that $Y\rightarrow X'$ happens.

The objective of the present paper is to assess whether two spatially separated parties can exploit the coherent superposition of two one-way signaling to locally implement larger set of quantum operations than LOCC on  their shared quantum system without using any  additional resource. To this end, let us design a protocol that spatially separated parties use to exploit the indefiniteness of causal structure as a resource.

\paragraph*{\textbf{Protocol:}}
Consider two parties Alice and Bob, spatially separated by a distance $h$, that share a bipartite quantum system in an arbitrary quantum state.  Alice and Bob have under their control not only their respective share of the system, but also additional ancillary systems. Importantly, there is no quantum correlation (like entanglement) between the ancillary systems of Alice and Bob that they can use as resource. Each of them has at their disposal two quantum operations $\lbrace \Phi_1,\Phi_2 \rbrace$ that they can locally implement on their respective subsystems  and/or ancillary systems. There is another agent Charlie, who is able to prepare mass configuration in spatial superposition of two locations -- one near Alice's laboratory, at a distance $R$ from it: ($\ket{\mathcal{M}_{Y\rightarrow X}}$), and the other near Bob's laboratory, at a distance $R$ from it: ($\ket{\mathcal{M}_{X\rightarrow Y}}$), i.e., in the state $(\ket{\mathcal{M}_{X\rightarrow Y}} + \ket{\mathcal{M}_{Y\rightarrow X}})/\sqrt{2}$. Furthermore, Charlie is also able to implement local quantum operations (including measurements) on the mass configuration and communicate classical information to both Alice and Bob. Alice and Bob will implement one quantum operation chosen from the above set at their respective local time $\tau^*$ (i.e., $t_A=\tau^*$ and $t_B=\tau^*$), where $\tau^*$ satisfies Eq. (\ref{tau*}) with $r_A=R$. Which particular quantum operation they will implement on their subsystem depends on the following strategy. If Alice (Bob) receives a light signal from Bob (Alice) before her (his) local time $\tau^*$, she (he) will implement operation $\Phi_2$ on  systems (shared subsystem and/or ancillary system) under her (his) control. Otherwise, Alice (Bob) will implement operation $\Phi_1$ on  systems under her (his) control and send a light signal (classical communication) to Bob (Alice)  (see Fig. \ref{schematic}). After implementing quantum operations by both the parties, Charlie will measure the mass configuration in the basis -- $\{ (\ket{\mathcal{M}_{X\rightarrow Y}} \pm \ket{\mathcal{M}_{Y\rightarrow X}})/\sqrt{2}\}$ -- and communicate his outcome to Alice and Bob. 

\begin{figure}
    \centering
    \includegraphics[height=240px, width=220px]{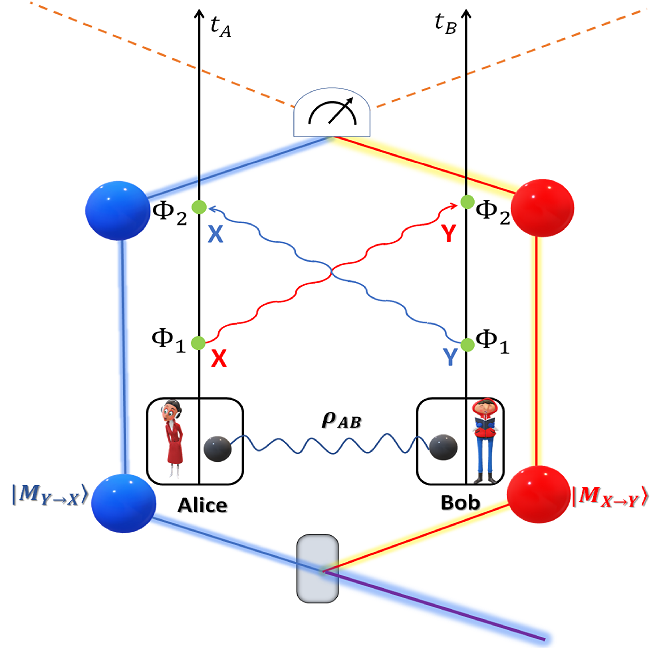}
    \caption{Schematics of the protocol for exploiting indefinite causal structure as resource for local implementation of non-local operations. A gravitating mass is prepared in a spatial superposition of two distinct locations with spatial states being denoted by $|\mathcal{M}_{X\rightarrow Y}\rangle$ (red) and $|\mathcal{M}_{Y\rightarrow X}\rangle$ (blue). Local proper time $t_A=\tau^*$ at Alice's laboratory and $t_B=\tau^*$ at Bob's laboratory define event $X$ and $Y$ respectively. Alice and Bob share an unknown bipartite quantum state. If the mass is in $|\mathcal{M}_{X\rightarrow Y}\rangle$ state, event $X$ lies in causal past of event $Y$ (represented in red). In that case, Alice implements local operation $\Phi_{1}$ at event $X$ on her subsystem and send signal to Bob, who implements operation $\Phi_{2}$ on his subsystem at event $Y$. Whereas, if the mass is in $|\mathcal{M}_{Y\rightarrow X}\rangle$ state, event $Y$ lies in causal past of event $X$ (represented in blue). In that case, Bob implements local operation $\Phi_{1}$ at event $Y$ on his subsystem and send signal to Alice, who implements operation $\Phi_{2}$ on her subsystem at event $Y$. Finally, the mass is measured in $\{ (\ket{\mathcal{M}_{X\rightarrow Y}} \pm \ket{\mathcal{M}_{Y\rightarrow X}})/\sqrt{2}\}$ basis, and the outcome is broadcast to any event at Alice's and Bob's laboratory that lie in the future light-cone of the event at which the measurement is done.}
    \label{schematic}
\end{figure}

Before proceeding further, let us review the roles of Alice, Bob, and Charlie and their allowed operations. As mentioned earlier, Alice and Bob are sharing a single copy of a bipartite quantum system whose state is unknown to them. They are allowed to implement only local quantum operations on the states of their respective subsystem and/or locally accessible ancillary systems. The particular local operations $\{\Phi_1,\Phi_2\}$ that they can choose are to be implemented only at the assigned times, i.e., at respective local times $t_A=\tau^*$ and $t_B=\tau^*$. Additionally, they can exchange classical information with each other. But they do not possess any other additional quantum resource that they can resort to. Charlie is located in a far-away separate laboratory and he is allowed to implement operations on the mass configuration only. These include creating spatial superposition of the mass, measuring and implementing unitary transformation on it. Importantly, even if Charlie's operations are nonlocal, he is restricted to access the mass configuration only and, in particular, he is not allowed to access Alice's and Bob's laboratories. Also, Charlie can communicate only classical information to Alice and Bob. 

However, since Charlie can create spatial superposition of the mass configuration, he is able to generate the resource of indefinite causal structure. Therefore, his allowed operations are ``resource-generating". This is not surprising because Charlie is able to implement nonlocal operations. Furthermore, Charlie's measurement in the $\{(\ket{\mathcal{M}_{X\rightarrow Y}} \pm \ket{\mathcal{M}_{Y\rightarrow X}})/\sqrt{2}\}$ basis collapses the mass configuration in a spatially superposed state that again gives rise to indefinite causal structure. Thus, once Charlie generates the resource by creating spatial superposition of the mass, his measurement preserves the amount of generated resource.

\section{Indefinite causal structure as a universal resource}\label{four} 

At first, we show that Alice and Bob are able to implement entangling operation on the shared quantum systems using the aforementioned protocol in the presence of indefinite causal structure, i.e., when Charlie prepares mass configuration in spatial superposition of two locations, which otherwise is not possible by LOCC alone. This can be shown as follows. Let Alice and Bob share a bipartite quantum system in product state, $\ket{\psi}_A\otimes\ket{\phi}_B \in \mathbb{C}^d_A \otimes \mathbb{C}^d_B $ and the quantum operations that they can locally implement on their subsystems be two \textit{unitary operations} $\lbrace U_1,U_2 \rbrace$. Now, when Alice's event $X$, defined by her local time $t_A=\tau^*$, is in the causal past of Bob's event $Y$, defined his local time $t_B=\tau^*$, she will not receive any signal from Bob. Therefore, she will apply $U_1$ on her subsystem $\ket{\psi}_A$ and will send a signal to Bob. Bob, upon receiving the signal from Alice, will apply $U_2$ on his subsystem $\ket{\phi}_B$. On the other hand, when event $Y$ is in the causal past of event $X$, then Bob will not receive any signal from Alice. Consequently, he will apply $U_1$ on his subsystem $\ket{\phi}_B$ and will send a signal to Alice who will then apply $U_2$ on her subsystem $\ket{\phi}_B$. Therefore, the joint state of Alice's and Bob's subsystems and the mass configuration of Charlie in the presence of indefinite causal structure can be represented as
\begin{small}
\begin{equation}
\label{two_event_entanglement}
\frac{1}{\sqrt{2}}\left( |\mathcal{M}_{X\rightarrow Y}\rangle U_1\ket{\psi}_A U_2\ket{\phi}_B + |\mathcal{M}_{Y\rightarrow X}\rangle U_2\ket{\psi}_A U_1\ket{\phi}_B\right).
\end{equation}
\end{small}
Finally, Charlie will perform a measurement on the mass configuration in $\left\lbrace \left( |\mathcal{M}_{X\rightarrow Y}\rangle \pm |\mathcal{M}_{Y\rightarrow X}\rangle \right)/\sqrt{2}\right\rbrace$ basis and communicate his outcome to Alice and Bob. Now, depending on the outcome of Charlie's measurement, the joint state of Alice and Bob will collapse on one of the following states,
\begin{equation}
\frac{1}{\sqrt{2}}\left(U_1\ket{\psi}_A U_2\ket{\phi}_B \pm U_2\ket{\psi}_A U_1\ket{\phi}_B\right),
\label{entanglementu1u2}
\end{equation}
which is an entangled state for suitable choices of $U_1$ and $U_2$. It is evident from Eq. \eqref{two_event_entanglement} that communicating the measurement outcome obtained by Charlie to Alice and Bob is necessary because tracing out Charlie's part will never lead to entanglement between Alice and Bob, rather a classical statistical mixture of the shared states (after the action of local unitaries) will be generated. Note here that the two unitary operators are not necessarily noncommuting in order to utilise indefinite causal order as useful resource for generating entanglement under LOCC.

Next, we show that any quantum operation (not only entangling operation) on an arbitrary bipartite state $ \rho_{AB} \in D(\mathbb{C}^2\otimes\mathbb{C}^d)$ can be implemented under LOCC in the presence of indefinite causal structure by exploiting the aforementioned protocol.

\begin{theorem}
$2$-ICS (indefinite causal structure yielding indefinite causal order of two events) is a \textbf{universal} resource for local implementation of the whole set of quantum operations $\mathcal{O}$ on an arbitrary bipartite state $ \rho_{AB} \in  D(\mathbb{C}^2\otimes\mathbb{C}^d)$ shared by two spatially separated parties.
\end{theorem}
\begin{proof}
We begin by considering that the two spatially separated parties, Alice and Bob, share a bipartite quantum system in an  arbitrary pure quantum state given by,
\begin{align}
    \ket{\Psi}_{AB} = \sum_{i=0}^1\sum_{j=0}^{d-1} \alpha_{ij} \ket{i}_A \ket{j}_B \in \mathbb{C}^2\otimes\mathbb{C}^d,
\end{align}
where $\alpha_{ij}$ are, in general, complex numbers; $\sum_{i=0}^1\sum_{j=0}^{d-1} |\alpha_{ij}|^2=1$ and $\lbrace \ket{i}_A\rbrace_{i=0}^1$, and $\lbrace \ket{j}_B\rbrace_{j=0}^{d-1}$ are orthonormal bases in $\mathbb{C}^2$ and $\mathbb{C}^d$ respectively. In addition,
Bob possess another ancillary qubit $\ket{0}_{B'}$. Therefore, the joint state of the quantum system, shared between Alice and Bob and Bob's ancillary system is given by,
\begin{widetext}
\begin{align}
    \ket{\Psi}_{AB}\otimes\ket{0}_{B'} = \sum_{j=0}^{d-1} \left( \alpha_{0j} \ket{0}_A \ket{j}_B \ket{0}_{B'} + \alpha_{1j} \ket{1}_A \ket{j}_B \ket{0}_{B'} \right).
\end{align}
Now, Alice and Bob will perform the protocol (discussed earlier) with $U_1=\mathbb{I}$ and $U_2=\sigma_x$ on the subsystems $A$ and $B'$ respectively. After implementing the local unitary operations, the joint state of the shared quantum system, Bob's ancilla, and the mass configuration of Charlie can be written as
\begin{align}
\frac{1}{\sqrt{2}} &\left[ |\mathcal{M}_{X\rightarrow Y}\rangle \sum_{j=0}^{d-1} \left(\alpha_{0j} \ket{0}_A \ket{j}_B \ket{1}_{B'} + \alpha_{1j} \ket{1}_A \ket{j}_B \ket{1}_{B'}\right)\right.\left.+ |\mathcal{M}_{Y\rightarrow X}\rangle \sum_{j=0}^{d-1} \left( \alpha_{0j} \ket{1}_A \ket{j}_B \ket{0}_{B'} + \alpha_{1j} \ket{0}_A \ket{j}_B \ket{0}_{B'}\right)\right]
\end{align}
Charlie, after measuring the mass configuration in $\left\lbrace \left( |\mathcal{M}_{X\rightarrow Y}\rangle \pm |\mathcal{M}_{Y\rightarrow X}\rangle \right) /\sqrt{2}\right\rbrace$ basis communicates the outcome $``+"$ or $``-"$ to Bob classically. Accordingly, the joint state of the shared system and Bob's ancilla collapses into one of the following states: 
\begin{align}
\label{postmeasurementAB}
&\frac{1}{\sqrt{2}} \left[ \sum_{j=0}^{d-1} \left(\alpha_{0j} \ket{0}_A \ket{j}_B \ket{1}_{B'} + \alpha_{1j} \ket{1}_A \ket{j}_B \ket{1}_{B'}\right)\pm \sum_{j=0}^{d-1} \left( \alpha_{0j} \ket{1}_A \ket{j}_B \ket{0}_{B'} + \alpha_{1j} \ket{0}_A \ket{j}_B \ket{0}_{B'}\right)\right].
\end{align} 
\end{widetext}

\begin{table*}[t]
   { \centering
\begin{tabular}{||c|c|c|c|c||} 
 \hline
 Charlie's outcome & Alice's outcome & Bob's collapsed state & Bob's correction operation & Bob's final state \\
[0.5ex]
 \hline\hline
 $+$ & $0$ & $\sum\limits_{j=0}^{d-1} \left(\alpha_{0j}\ket{1}_{B'}\ket{j}_B+\alpha_{1j}\ket{0}_{B'}\ket{j}_B \right)$ & $\sigma_x\otimes\mathbb{I}$ &  \multirow{4}{*}{$\sum\limits_{i=0}^1\sum\limits_{j=0}^{d-1} \alpha_{ij} \ket{i}_{B'} \ket{j}_B$} \\
  \cline{1-4}
 $+$ & $1$ & $\sum\limits_{j=0}^{d-1} \left( \alpha_{0j}\ket{0}_{B'}\ket{j}_B+\alpha_{1j}\ket{1}_{B'}\ket{j}_B \right)$ & $\mathbb{I}\otimes\mathbb{I}$ &  \\ 
 \cline{1-4}
 $-$ & $0$ & $\sum\limits_{j=0}^{d-1} \left(\alpha_{0j}\ket{1}_{B'}\ket{j}_B-\alpha_{1j}\ket{0}_{B'}\ket{j}_B \right)$ & $i\sigma_y\otimes\mathbb{I}$ & \\ \cline{1-4}
 $-$ & $1$ & $\sum\limits_{j=0}^{d-1} \left(\alpha_{0j}\ket{0}_{B'}\ket{j}_B-\alpha_{1j}\ket{1}_{B'}\ket{j}_B \right)$ & $\sigma_z\otimes\mathbb{I}$ & \\
 \hline
\end{tabular}
\caption{Details of the teleportation protocol from Alice to Bob using $2$-ICS for locally implementing any quantum operation on an arbitrary pure quantum state $|\Psi_{AB} \rangle \in  \mathbb{C}^2\otimes\mathbb{C}^d$ shared by Alice and Bob.}
    \label{tab1} }
\end{table*}

After the completion of the above protocol, Alice measures her qubit in $\lbrace |0\rangle_A,|1\rangle_A \rbrace$ basis and communicates her outcome $k \in \lbrace 0,1\rbrace$ classically to Bob. Here, it might be noted that in order to communicate Alice's outcome to Bob, it is necessary to have an event defined in Alice's laboratory which is definitely in the causal past of an event defined in Bob's laboratory even when the causal structure of the space-time is indefinite. The possibility of such events are discussed in Appendix \ref{appa}. Bob then accordingly makes a correction operation on his ancillary qubit to get $\sum_{i=0}^1\sum_{j=0}^{d-1} \alpha_{ij} \ket{i}_{B'} \ket{j}_B$ as mentioned in the Table \ref{tab1}. 

This completes the teleportation process from Alice to Bob, i.e., the joint state of the shared system between Alice and Bob, and Bob's ancillary qubit has transformed as $\ket{\Psi}_{AB}\otimes\ket{0}_{B'}\mapsto \ket{0}_{A}\otimes\ket{\Psi}_{B'B}$ or, $\ket{\Psi}_{AB}\otimes\ket{0}_{B'}\mapsto \ket{1}_{A}\otimes\ket{\Psi}_{B'B}$. 

Since, it is possible to create an arbitrary pure state initially shared by Alice-Bob in Bob's local laboratory with unit  fidelity, it is also possible to create in Bob's laboratory any mixed state $\rho_{AB} = \sum_l p_l \ket{\Psi^l}\bra{\Psi^l}_{AB} \in D(\mathbb{C}^2 \otimes \mathbb{C}^d)$ (where $0 \leq p_l \leq 1$; $\sum_l p_l = 1$, $\{\ket{\Psi^l} \}$ is an orthonormal basis in $\mathbb{C}^2 \otimes \mathbb{C}^d$) initially shared by Alice-Bob with unit  fidelity using this protocol:
\begin{align}
    \rho_{AB}\otimes \ket{0}\bra{0}_{B'} &= \left(\sum_l p_l \ket{\Psi^l}\bra{\Psi^l}_{AB}\right) \otimes \ket{0}\bra{0}_{B'} \nonumber\\
    &\mapsto \ket{\xi}\bra{\xi}_{A} \otimes\left(\sum_l p_l \ket{\Psi^l}\bra{\Psi^l}_{B'B}\right), \nonumber \\
    & = \ket{\xi}\bra{\xi}_{A} \otimes\rho_{B' B}
\end{align} 
where $\ket{\xi} \in\lbrace\ket{0},\ket{1}\rbrace$. 

Therefore, as a consequence of the above protocol, Bob can locally implement any quantum operation from the set $\mathcal{O}$ on $\rho_{B'B}$, transforming it to $\sigma_{B'B}$ (say). 

As noted earlier, after Charlie's measurement, the postmeasurement state of the mass configuration is either $\left( |\mathcal{M}_{X\rightarrow Y}\rangle + |\mathcal{M}_{Y\rightarrow X}\rangle \right)/\sqrt{2}$ or, $\left( |\mathcal{M}_{X\rightarrow Y}\rangle - |\mathcal{M}_{Y\rightarrow X}\rangle \right)/\sqrt{2}$, both of which give rise to $2$-ICS and, hence, can be used again as a resource to teleport the transformed state of Bob's ancillary qubit to Alice's laboratory in such a way that the joint state $\sigma_{B'B}$ belonging to Bob's laboratory now becomes the shared state between Alice and Bob. To elaborate it specifically, for simplicity, we may consider that Charlie resets the mass configuration in the state: $\left( |\mathcal{M}_{X\rightarrow Y}\rangle + |\mathcal{M}_{Y\rightarrow X}\rangle \right)/\sqrt{2}$ by local unitary operation, and Alice also resets her qubit  to $\ket{0}_A$ (i.e., $\ket{\xi}_A \mapsto \ket{0}_A$) by local unitary operation. Now, by following the same protocol as mentioned earlier with a different pair of events and interchanging Alice's and Bob's roles, it is possible to realise this back-teleportation:
\begin{align}
    \ket{0}\bra{0}_A\otimes\sigma_{B'B} \mapsto \sigma_{AB}\otimes \ket{\kappa}\bra{\kappa}_{B'},
\end{align}
where $\ket{\kappa} \in \lbrace \ket{0},\ket{1} \rbrace$.

In this way, the two spatially separated parties can locally implement any  quantum operation $\Phi \in \mathcal{O}$ on their shared system transforming $\rho_{AB}\mapsto \sigma_{AB}$. This completes the proof.
\end{proof}

\begin{remark}
\label{independentresource}
In the above teleportation protocol, $2$-ICS serves as an independent resource from maximally entangled two-qubit states that act as a necessary resource for teleportation of qubits following the standard protocol \cite{Bennett1993}.
\end{remark}

In other words, it is not the case that the $2$-ICS is used in our protocol to generate a maximally entangled two-qubit state (as in Eq.(\ref{entanglementu1u2}) with some specific $U_1$ and $U_2$) such that this entanglement can be used  as a resource to teleport a third qubit from one location to another. Importantly, no maximally entangled state appears as an intermediate step in our proposed teleportation protocol with unit fidelity.

The fact that $2$-ICS can be used as a resource to teleport qubits opens up the possibilities for local execution of other nonlocal tasks involving more than two spatially separated parties as well, using $2$-ICS as resource.

\begin{corollary}\label{corollary1}
The set of orthogonal product states in $\mathbb{C}^2\otimes\mathbb{C}^2\otimes\mathbb{C}^d$ (introduced in \cite{Bennet99}) exhibiting ``quantum nonlocality without entanglement'' can be perfectly distinguished locally with the aid of $2$-ICS as resource.
\end{corollary}
\begin{proof}
Let us consider that three spatially separated parties, say, Alice, Bob and Charlie share the product state $|\psi\rangle_{ABC} = |\psi_1\rangle_A \otimes |\psi_2\rangle_B \otimes |\psi_3\rangle_C \in \mathbb{C}^2\otimes\mathbb{C}^2\otimes\mathbb{C}^d$. Let us also assume that Charlie has an ancillary qubit $|0\rangle_{C'}$. Now, following our proposed teleportation  protocol from Alice to Charlie, the following transformation is possible:
\begin{align}
    |\psi\rangle_{ABC} \otimes &|0\rangle_{C'} \mapsto |\xi\rangle_A \otimes |\psi\rangle_{BCC'},\nonumber \\
    &\text{where} \, \, |\psi\rangle_{ABC} = |\psi_1\rangle_A \otimes |\psi_2\rangle_B \otimes |\psi_3\rangle_C \nonumber \\
    &\text{and} \, \, |\psi\rangle_{BCC'} =  |\psi_2\rangle_B \otimes |\psi_3\rangle_C \otimes |\psi_1\rangle_{C'}.
\end{align}
Hence, starting from a set of orthogonal product states in $\mathbb{C}^2\otimes\mathbb{C}^2\otimes\mathbb{C}^d$ shared between Alice, Bob and Charlie, Bob and Charlie will end up in sharing a set of orthogonal product states in $\mathbb{C}^2\otimes\mathbb{C}^{d+2}$. Therefore, the given task becomes the task of locally distinguishing a set of orthogonal bipartite product states belonging in $ \mathbb{C}^2\otimes\mathbb{C}^{d+2}$, which is possible as any set of orthogonal product states in $\mathbb{C}^2\otimes\mathbb{C}^{d'}$ can always be locally discriminated without any resource \cite{Bennett99(1),DiVincenzo03}. Note that here only two parties are required to access the resource of indefinite causal structure.
\end{proof}

The status of $2$-ICS as a universal resource for local implementation of nonlocal quantum operations on $\mathbb{C}^2\otimes\mathbb{C}^d$ systems has been sufficiently established by showing that it can be used as resource for back and forth perfect teleportation  between two distant laboratories. However, teleportation is not always necessary for execution of nonlocal tasks using indefinite causal structure. We illustrate this in the following example by showing that even without invoking teleportation two spatially separated parties can locally distinguish the four Bell states using $2$-ICS as resource.

Consider the set of four Bell states:
\begin{align}
    |\mathcal{B}_1\rangle &= \frac{1}{\sqrt{2}}(|0\rangle_A|0\rangle_B+|1\rangle_A|1\rangle_B), \nonumber \\
    |\mathcal{B}_2\rangle &= \frac{1}{\sqrt{2}}(|0\rangle_A|0\rangle_B-|1\rangle_A|1\rangle_B), \nonumber \\
    |\mathcal{B}_3\rangle &= \frac{1}{\sqrt{2}}(|0\rangle_A|1\rangle_B+|1\rangle_A|0\rangle_B), \nonumber \\
    |\mathcal{B}_4\rangle &= \frac{1}{\sqrt{2}}(|0\rangle_A|1\rangle_B-|1\rangle_A|0\rangle_B).\nonumber
\end{align}
Suppose a state is randomly chosen from the above set and distributed between two spatially separated parties, Alice and Bob. Their task is to perfectly identify which particular state is given to them. It is not possible for them to locally discriminate the given state by performing LOCC without the aid of any other resource \cite{PhysRevLett.87.277902}. However, if an additional maximally entangled two-qubit state is distributed between them as a resource, then the task can be accomplished perfectly. This resource has been shown to be both necessary as well as sufficient for the perfect execution of the task \cite{Bandyopadhyay2015ieee,HorodeckiSen2003}. Here we show that $2$-ICS is an equally useful resource for the task of perfect local discrimination of four Bell states.

\begin{prop}\label{propo1}
Perfect local discrimination of the set of four Bell states: $\left\lbrace |\mathcal{B}_i\rangle\right\rbrace_{i=1}^4$, shared between two spatially separated laboratories is possible by LOCC when $2$-ICS is used as the sole resource.
\end{prop}
\begin{proof}
The proof is straightforward.
Alice and Bob will follow the earlier discussed protocol with local unitary operations $U_1=\mathbb{I}$ and $U_2=\sigma_x$  on their respective qubit of the shared Bell state. Note here that no ancillary qubit is required in this task. Now, depending on the given state between Alice and Bob, the joint state of Alice's and Bob's subsystems and the mass configuration of Charlie transforms as
\begin{widetext}
\begin{small}
\begin{align}
    \frac{1}{\sqrt{2}} \left( |\mathcal{M}_{X\rightarrow Y}\rangle + |\mathcal{M}_{Y\rightarrow X}\rangle \right) \otimes |\mathcal{B}_1\rangle &\rightarrow \frac{1}{\sqrt{2}} \left( |\mathcal{M}_{X\rightarrow Y}\rangle \otimes\frac{1}{\sqrt{2}}(|0\rangle_A|1\rangle_B+|1\rangle_A|0\rangle_B) + |\mathcal{M}_{Y\rightarrow X}\rangle \otimes \frac{1}{\sqrt{2}}(|1\rangle_A|0\rangle_B+|0\rangle_A|1\rangle_B) \right)  \nonumber\\
    &=\frac{1}{\sqrt{2}} \left( |\mathcal{M}_{X\rightarrow Y}\rangle + |\mathcal{M}_{Y\rightarrow X}\rangle \right) \otimes |\mathcal{B}_3\rangle, \nonumber\\
    \frac{1}{\sqrt{2}} \left( |\mathcal{M}_{X\rightarrow Y}\rangle + |\mathcal{M}_{Y\rightarrow X}\rangle \right) \otimes |\mathcal{B}_2\rangle &\rightarrow \frac{1}{\sqrt{2}} \left( |\mathcal{M}_{X\rightarrow Y}\rangle \otimes \frac{1}{\sqrt{2}}(|0\rangle_A|1\rangle_B-|1\rangle_A|0\rangle_B) + |\mathcal{M}_{Y\rightarrow X}\rangle\otimes \frac{1}{\sqrt{2}}(|1\rangle_A|0\rangle_B-|0\rangle_A|1\rangle_B) \right)  \nonumber\\
    &=\frac{1}{\sqrt{2}} \left( |\mathcal{M}_{X\rightarrow Y}\rangle - |\mathcal{M}_{Y\rightarrow X}\rangle \right) \otimes |\mathcal{B}_4\rangle, \nonumber\\
    \frac{1}{\sqrt{2}} \left( |\mathcal{M}_{X\rightarrow Y}\rangle + |\mathcal{M}_{Y\rightarrow X}\rangle \right) \otimes |\mathcal{B}_3\rangle &\rightarrow \frac{1}{\sqrt{2}} \left( |\mathcal{M}_{X\rightarrow Y}\rangle \otimes \frac{1}{\sqrt{2}}(|0\rangle_A|0\rangle_B+|1\rangle_A|1\rangle_B) + |\mathcal{M}_{Y\rightarrow X}\rangle \otimes \frac{1}{\sqrt{2}}(|1\rangle_A|1\rangle_B)+|0\rangle_A|0\rangle_B) \right) \nonumber\\
    &=\frac{1}{\sqrt{2}} \left( |\mathcal{M}_{X\rightarrow Y}\rangle + |\mathcal{M}_{Y\rightarrow X}\rangle \right) \otimes |\mathcal{B}_1\rangle, \nonumber\\
    \frac{1}{\sqrt{2}} \left( |\mathcal{M}_{X\rightarrow Y}\rangle + |\mathcal{M}_{Y\rightarrow X}\rangle \right) \otimes |\mathcal{B}_4\rangle &\rightarrow \frac{1}{\sqrt{2}} \left( |\mathcal{M}_{X\rightarrow Y}\rangle\otimes \frac{1}{\sqrt{2}}(|0\rangle_A|0\rangle_B-|1\rangle_A|1\rangle_B) + |\mathcal{M}_{Y\rightarrow X}\rangle \otimes \frac{1}{\sqrt{2}}(|1\rangle_A|1\rangle_B-|0\rangle_A|0\rangle_B)\right)  \nonumber\\
     &=\frac{1}{\sqrt{2}} \left( |\mathcal{M}_{X\rightarrow Y}\rangle - |\mathcal{M}_{Y\rightarrow X}\rangle \right) \otimes |\mathcal{B}_2\rangle.
\end{align}
\end{small}
\end{widetext}
Thereafter, Charlie will measure the mass configuration in the basis: $\left\lbrace \left( |\mathcal{M}_{X\rightarrow Y}\rangle \pm |\mathcal{M}_{Y\rightarrow X}\rangle \right)/\sqrt{2}\right\rbrace$ and communicate his outcome to Alice and Bob. Alice and Bob will then measure their respective qubits in the basis: $\lbrace |0\rangle, |1\rangle \rbrace$ and communicate their outcomes to each other. For Alice to communicate her outcome to Bob, an event defined at her laboratory must be definitely in the causal past of an event defined at Bob's laboratory, even in space-time with indefinite causal structure. Similarly, for Bob to communicate his outcome to Alice, an event defined at his laboratory must be definitely in the causal past of an event defined at Alice's laboratory. As mentioned earlier, this is always possible (See Appendix \ref{appa}). From all the information gathered, Alice and Bob can determine which state is shared between them as mentioned in the Table \ref{tab2}. This ends the proof.
\end{proof}

Note that the successful execution of this nonlocal task using indefinite causal structure as resource does not require any kind of teleportation from one laboratory to the other. Furthermore, no new entangled state is generated at any step of the above protocol. This is crucial as it clearly marks the independence of indefinite causal structure as a resource from entanglement, as already mentioned in Remark \ref{independentresource}.

As a corollary of the above proposition, another important result follows, illustrating the scope of indefinite causal structure as resource for nonlocal task involving more than two spatially separated parties.
\begin{corollary}[Unlocking of Smolin state]
Consider four parties, Alice, Bob, Dan, and Emma share a state, $\rho_{ABDE}=\frac{1}{4}\sum_{i=1}^4 |\mathcal{B}_i\rangle\langle \mathcal{B}_i|_{AB}\otimes|\mathcal{B}_i\rangle\langle \mathcal{B}_i|_{DE}$ \cite{Smolin2001}. They can distill entanglement from this bound entangled (Smolin) state under LOCC using $2$-ICS as a sole resource.
\end{corollary}

\begin{proof}
Since the four parties share the state $\rho_{ABDE}$, Alice and Bob share one of the four Bell states $\lbrace |\mathcal{B}_i\rangle\rbrace$, but do not know which one; Dan and Emma share the same Bell state, also not knowing which one it is. Since Alice and Bob can locally determine which Bell state they are sharing using $2$-ICS, they can communicate this classical information to Dan and Emma who will then know which Bell state they have and can locally convert it to any desired Bell state \cite{Smolin2001}, distilling entanglement from the Smolin state. After the distillation, each of the pairs -- Alice-Bob and Dan-Emma -- share 1 ebit of entanglement between them.
\end{proof}

\begin{table}[t]
{\centering
\begin{tabular}{||c|c|c|c||} 
 \hline
 Charlie's & Alice's  & Bob's  & Conclusion \\ [0.5ex] 
  outcome & outcome & outcome & \\ [0.5ex] 
 \hline\hline
 $+$ & $0$ & $1$ & $|\mathcal{B}_1\rangle$\\ 
 \hline
 $+$ & $1$ & $0$ & $|\mathcal{B}_1\rangle$ \\
 \hline
 $+$ & $0$ & $0$ & $|\mathcal{B}_3\rangle$\\ 
 \hline
 $+$ & $1$ & $1$ & $|\mathcal{B}_3\rangle$ \\
 \hline
 $-$ & $0$ & $1$ & $|\mathcal{B}_2\rangle$\\ 
 \hline
 $-$ & $1$ & $0$ & $|\mathcal{B}_2\rangle$ \\
 \hline
 $-$ & $0$ & $0$ & $|\mathcal{B}_4\rangle$\\ 
 \hline
 $-$ & $1$ & $1$ & $|\mathcal{B}_4\rangle$ \\
 \hline
\end{tabular}
\caption{Protocol for discriminating four Bell states locally using $2$-ICS}
    \label{tab2} }
\end{table}

\paragraph*{\textbf{Generalisation to higher dimensional bipartite systems:}}
While $2$-ICS serves as universal resource for local implementation of nonlocal quantum operations on an arbitrary $2 \times d$-dimensional bipartite quantum system,  
$3$-ICS can be shown to be a universal resource for local implementation of nonlocal quantum operations on any $3 \times d$-dimensional bipartite quantum system. Similarly, $4$-ICS can be shown to serve as universal resource for local implementation of nonlocal quantum operations on any $4 \times d$-dimensional bipartite quantum system. A more detailed discussion is given in Appendix \ref{gen}. From these results, we conjecture  that $d'$-ICS is a universal resource for local implementation of nonlocal quantum operations on an arbitrary bipartite state $\rho \in D(\mathbb{C}^{d'}\otimes\mathbb{C}^d)$, where $d'\leq d$.

\section{\label{five} Conclusion and Discussion}
Searching for new quantum  resources and finding out their information theoretic applications have attracted much attention in recent years. In the present article, we have shown that indefinite causal structure acts as a universal resource for implementing nonlocal quantum operations as well as nonlocal quantum tasks under LOCC on quantum systems distributed among two spatially separated parties. Specifically, we have proved it by proposing a novel back and forth perfect teleportation protocol using indefinite causal structure as the sole resource, which is completely different from the standard teleportation protocol that uses maximally entangled state as resource \cite{Bennett1993}. This proposed teleportation protocol may be helpful in gaining deeper understanding of the power of indefinite causal order. Consequently, we have also shown that the set of product states in $\mathbb{C}^2\otimes\mathbb{C}^2\otimes\mathbb{C}^d$ exhibiting quantum nonlocality without entanglement \cite{Bennet99} can also be perfectly distinguished under LOCC following our protocol utilizing indefinite causal structure solely. Furthermore, by presenting another example of nonlocal tasks, namely, perfect local discrimination of the set of four Bell states, we have demonstrated that the above-mentioned teleportation protocol proposed by us is not always necessary for local execution of nonlocal tasks with the aid of that resource. As an immediate corollary of it, we have also shown that one can distill entanglement from the Smolin state \cite{Smolin2001} under LOCC using indefinite causal structure. It might be noted here that in our proposed protocol, the resource 
gets regenerated at the end of the execution of nonlocal quantum operations or tasks and, therefore, can be reused for subsequent tasks.

It is quite exemplary to discuss about the schematics of practical realisation of the protocol mentioned in Sec. \ref{three}. The possibility of creating spatial superposition of mass can be ascertained by sending a massive particle with an embedded electronic spin (spin-$\frac{1}{2}$) through a Stern-Gerlach (SG) interferometer. If the embedded spin is prepared in, say, $\frac{1}{\sqrt{2}}(\ket{\uparrow_z}+\ket{\downarrow_z})$ state (where $\ket{\uparrow_z}$ and $\ket{\downarrow_z}$ are the eigenstates of $\sigma_z$ with eigenvalues $+1$ and $-1$ respectively) and an inhomogenous magnetic field in the SG apparatus is applied in the $z$-direction, then the mass will be split into two paths depending on the spin, thereby creating a spatial superposition of the mass. Mathematically, if we represent the spatial state of the mass localised at the centre of axis of the SG apparatus before entering as $\ket{C}$, then the splitting action of the SG is given by, 
$\ket{C}\frac{1}{\sqrt{2}}(\ket{\uparrow_z}+\ket{\downarrow_z}) \mapsto \frac{1}{\sqrt{2}}(\ket{\mathcal{M}_{X\rightarrow Y}}\ket{\uparrow_z}+\ket{\mathcal{M}_{Y\rightarrow X}}\ket{\downarrow_z})$, where $\ket{\mathcal{M}_{X\rightarrow Y}}$ and $\ket{\mathcal{M}_{Y\rightarrow X}}$ are the two spatial states corresponding to the two paths after the SG-split. This strategy of creating spatial superposition of massive objects has been discussed explicitly in Ref. \cite{Bose17}.  Note however that we can adapt the scheme proposed in \cite{Bose17} to a larger scale such that the mass can exist in spatial superposition of two distant locations i.e., not confined to a single laboratory.

Now, in order to execute our proposed protocol, this superposition is to be hold for a time, say, $t_{s}$ (by switching off the magnetic field of the SG apparatus for that duration). During this time window, all the steps in the protocol before Charlie's measurement is to be carried out. At this stage, the joint state of the mass configuration and Alice's and Bob's subsystems (Eq. (\ref{two_event_entanglement})) becomes $\frac{1}{\sqrt{2}}(\ket{\mathcal{M}_{X\rightarrow Y},\uparrow_z}U_1\ket{\psi}_AU_2\ket{\phi}_B + \ket{\mathcal{M}_{Y\rightarrow X},\downarrow_z}U_2\ket{\psi}_AU_1\ket{\phi}_B).$ The paths are then brought together in a common location and refocused by a refocussing SG apparatus with
magnetic-field inhomogeneity oriented in the opposite direction ($-z$-direction) \cite{Bose17}, which is represented mathematically as $\ket{\mathcal{M}_{X\rightarrow Y},\uparrow_z}\mapsto \ket{C}\ket{\uparrow_z}$ and $\ket{\mathcal{M}_{Y\rightarrow X},\downarrow_z}\mapsto \ket{C}\ket{\downarrow_z}$. Thereby, the joint state of the mass configuration and Alice's and Bob's subsystems becomes $\frac{1}{\sqrt{2}}\ket{C}(\ket{\uparrow_z}U_1\ket{\psi}_AU_2\ket{\phi}_B + \ket{\downarrow_z}U_2\ket{\psi}_AU_1\ket{\phi}_B)$. In this way the entanglement between the spatial degree of freedom of the mass and Alice and Bob's subsystems can be mapped onto the entanglement between the spin degree of freedom of the mass and Alice's and Bob's subsystems. Finally, by measuring the spin of the mass in $\sigma_x$ basis, measurement in the $\{\frac{1}{\sqrt{2}}(\ket{\mathcal{M}_{X\rightarrow Y}}\pm\ket{\mathcal{M}_{Y\rightarrow X}})\}$ basis can be implemented effectively.

Before concluding a few remarks are in order. In our described protocol, Alice and Bob are able to locally implement nonlocal operations essentially due to the fact that mass superposition is nonlocal, in the sense that it is a spatial superposition of the mass  at two distant locations, and Charlie can perform operations on this mass configuration. Alice and Bob implement local conditional operations based on the state of the mass configuration -- if the state of the mass is  $\ket{\mathcal{M}_{X\rightarrow Y}}$, then they implement one kind of operation, whereas if state of the mass is $\ket{\mathcal{M}_{Y\rightarrow X}}$, they implement another kind of operation. Therefore, our protocol seems merely like a controlled operation (like C-NOT or C-SWAP etc.) that leads to entanglement between Alice and Bob (after Charlie's measurement), without actually requiring indefinite causal order. To fully appreciate the importance of the role of indefinite causal structure, first of all note that such kind of controlled operations are non-local and, hence, cannot be realised by LOCC without the assistance of any other quantum resource. That is, if the control and the target(s) are located at different locations, then it is not possible to implement controlled operations by locally operating on the control and the target(s) separately and communicating classical information between them. To implement them locally quantum communication between control and target(s) is necessary and for that quantum channel, entangled state or, some other quantum entity must be shared between the control and the target(s). In our case, the mass is situated neither in Alice's laboratory, nor in Bob's laboratory -- both of the locations where the mass exists in superposition are outside Alice's and Bob's laboratories. Alice and Bob, being limited to local operations only, cannot access the state of the mass. To use the mass in spatially superposed configuration as a control for implementing the conditional local operations, quantum information about the state of the mass must somehow reach Alice and Bob. In our case, this information is carried by the gravitational field due to the mass. It is important to note here that as the mass is in spatial superposition, the gravitational field it produces is also quantum in nature (the space-time is in a superposition of two different metrics) \cite{PhysRevA.101.052110,Galley2022nogotheoremnatureof,arxiv.2202.03368}. In other words, the quantum gravitational field acts as a quantum channel that carries quantum information from the locations of the mass to Alice's and Bob's locations. The indefiniteness of causal structure is a particular manifestation of the quantum nature of the space-time (the gravitational field) that Alice and Bob utilise operationally. Thus, indefiniteness of the causal structure plays a vital role in using the mass configuration as control.

Recently, it has been shown that if three or more parties are allowed to exploit indefinite causal order, then local quantum operations assisted with classical processes can perfectly discriminate sets of quantum states in $\mathbb{C}^2\otimes\mathbb{C}^2\otimes\mathbb{C}^2$ exhibiting quantum nonlocality without entanglement \cite{Baumeler14,Baumeler16,Kunjwal22}. However, our result shown in Corollary \ref{corollary1} indicates that the set of states can be distinguished perfectly by LOCC when only two parties use indefinite causal structure ($2$-ICS) as a resource. Furthermore, Akibue {\it et al.} have shown that the set of local quantum operations connected by classical communication without any predefined causal order (which they denote by LOCC$^*$) is equivalent to the set of separable operations \cite{Murao17}. Notwithstanding, since we have already shown that indefinite causal structure is a universal resource for local implementation of any quantum operation, 
our protocol exploiting indefinite causal structure is supposed  to perform bigger set of operations than the separable operations. For example, even in $\mathbb{C}^2\otimes\mathbb{C}^2$, we have shown that our protocol can be used to discriminate the set of four Bell states perfectly which otherwise is not possible by separable operations and thus establishing the superiority of our protocol.

Our results open up the possibilities of several open questions. The teleportation protocol discussed in this paper utilizing indefinite causal structure as a resource can, in principle, be generalised for arbitrary higher dimensional quantum systems which we would like to keep open as a future direction for the interested readers. This has also the potential to show some fascinating applications in quantum network and quantum internet. Furthermore, in the multipartite scenario, there will be more than two possible causal orders between the parties and, thus, the indefiniteness in causal structure will be quite complex. It is thus interesting to generalize our study in multipartite context. Further, exploring the full quantum advantage of the indefinite causal structure in multipartite scenario may lead to new applications in information theory. Finally, quantifying the resources utilized in our study and connecting it with other  information theoretic resources are to be explored further. 

\begin{acknowledgments}
We thank Somshubhro Bandyopadhyay, Manik Banik and Mir Alimuddin for fruitful discussions and helpful insights.  We also thank the anonymous reviewers for their helpful  suggestions and illuminating discussion. P.G. acknowledges Department of Science $\&$ Technology, Government of India for financial support.  D.D. acknowledges the Royal Society (United Kingdom) for the support through the Newton International Fellowship (NIF$\backslash$R$1\backslash212007$). 
\end{acknowledgments}

\bibliography{ICO}
\begin{appendix}
\onecolumngrid
\section{\label{appa} Definite causal order in Indefinite causal structure}
In this section, we show that there always exist pairs of events whose causal order is definite even in space-time with indefinite causal structure. Let us go back to the illustration in Sec. \ref{sup}. Consider three events -- (i) $X$, defined by the local time $t_A=\tau^*$ at clock $A$, (ii) $Y$, defined by the local time $t_B=\tau^*$ at clock $B$, and (iii) $Y'$ defined by the local time $t_B=\tilde{\tau}$ at clock $B$, in their respective locations, such that $\tilde{\tau}>\tau^*$ and $\tau^*$ satisfies Eq. (\ref{tau*}) with $r_A=R$ and $r_B=R+h$. When the mass is located near the clock $B$ (i.e., quantum state of the mass is $\ket{\mathcal{M}_{X\rightarrow Y}}$), $t_B=\tau^*$ is in the causal future of $t_A=\tau^*$. Since we have considered $\tilde{\tau}>\tau^*$, $t_B=\tilde{\tau}$ will also be in causal future of $t_A=\tau^*$. That is event $Y'$ is in causal future of event $X$: $X\rightarrow Y'$.

On the other hand, when the mass is located near the clock $A$ (i.e., quantum state of the mass is $\ket{\mathcal{M}_{Y\rightarrow X}}$), if we want event $Y'$ to be in causal future of event $X$, light signal emitted from event $X$ must reach clock $B$'s location before $t_B=\tilde{\tau}$. The coordinate time taken for a light signal to travel from clock $A$ to clock $B$ is $T'_c=\frac{1}{c} \int_{R}^{R+h} dr' \sqrt{-\frac{g_{rr}(r')}{g_{tt}(r')}}$.
Time at clock $B$ when the light signal reaches there is $t_B=\sqrt{-g_{tt}(R+h)}\left( \frac{\tau_1}{\sqrt{-g_{tt}(R)}}+T'_c \right)$. This proper time must precede the proper time $t_B=\tilde{\tau}$, i.e., 
\begin{align}
\label{tau2}
    \sqrt{-g_{tt}(R+h)}\left( \frac{\tau^*}{\sqrt{-g_{tt}(R)}}+T'_c \right) \leq \tilde{\tau}.
\end{align}
Under this condition event $Y'$ will also be in the causal future of event $X$: $X\rightarrow Y'$. Similarly, we can also define events such that the event defined at the location of clock B is definitely in the causal past of the event defined at the location of clock A. So, if we choose the events $X$ and $Y'$ such that Eq. (\ref{tau2}) is satisfied along with Eq. (\ref{tau*}), we get definite causal order between them even in indefinite causal structure. This is shown in Figure \ref{ICO}.

\section{\label{gen} Indefinite causal structure as universal resource for local implementation of nonlocal quantum operations on higher dimensional quantum systems}

In this section we explicitly show how $3$-ICS and $4$-ICS can be used as universal resources for implementing nonlocal quantum operations locally on arbitrary quantum states in $\mathbb{C}^3\otimes\mathbb{C}^d (d\geq3)$ and $ \mathbb{C}^4\otimes\mathbb{C}^d (d\geq4)$ respectively. As mentioned in the main text, a sufficient way to prove this is by showing that $3(4)$-ICS can be used as the sole resource for back and forth teleportation of the state of the $3(4)$-dimensional subsystem of the initially shared system in such a way that the other party (that in possession of the $d$-dimensional subsystem) can have access to the entire state of the shared system in his/her local laboratory to implement any desirable quantum operation and finally can return the subsystem to its initial location. We show in the following that the said can be done for an arbitrary shared pure state and, hence, for any shared mixed state as well.

\subsubsection{$3$-ICS as universal resource for local implementation of nonlocal quantum operations on quantum states $\in \mathbb{C}^3\otimes\mathbb{C}^d$}

Consider the two spatially separated parties, Alice and Bob, share an arbitrary quantum state, $\ket{\psi}_{AB} = \sum_{i=0}^2 \sum_{j=0}^{d-1} \alpha_{ij} \ket{i}_A\ket{j}_B \in \mathbb{C}^3\otimes\mathbb{C}^d$ (where, $\sum_{i,j} |\alpha_{ij}|^2=1$) and Bob has an ancillary qutrit system in the state $|0\rangle_{B'}$ in his laboratory. Note that Bob's subsystem of the initially shared state will always be denoted by $B$ and the state of Bob's ancillary qutrit will be denoted by $B'$. Without loss of generality, let us define two events -- $X_1$ and $X_2$ in Alice's laboratory as Alice's local times $t_A=\tau_1$ and $t_A=\tau_2$ respectively with $\tau_1<\tau_2$, and another event $Y$ in Bob's laboratory as his local time $t_B=\tau_1$. Since, both the events $X_1$ and $X_2$ are in Alice's laboratory, their causal order is definite for any causal structure of the space-time, i.e., $X_1$ is always in causal past of $X_2$ ($X_1 \rightarrow X_2$). Depending on the geometry of the classical space-time, the events can have three classically distinct and mutually exclusive causal orders --\\
\\
(1) event $Y$ is in causal past of $X_1$ (which we denote $Y \rightarrow X_1 \rightarrow X_2$),\\
(2) event $Y$ is in causal future of $X_1$ but in causal past of $X_2$ (which we denote $X_1 \rightarrow Y \rightarrow X_2$), and\\
(3) event $Y$ is in causal future of $X_2$ (which we denote $X_1  \rightarrow X_2\rightarrow Y$).\\
\\
Let us denote the corresponding quantum states of mass configurations, each yielding one of the three space-time geometries with the above causal orders, as $\ket{\mathcal{M}_{Y \rightarrow X_1 \rightarrow X_2}}$, $\ket{\mathcal{M}_{X_1 \rightarrow Y \rightarrow X_2}}$ and $\ket{\mathcal{M}_{X_1  \rightarrow X_2\rightarrow Y}}$. Now, if we consider that the mass configuration exists in a coherent superposition of these three quantum states then space-time with the three distinct geometries (each resulting in one of the three distinct causal orders) also exist in superposition. We consider that a third agent, Charlie, can prepare the mass configuration in the quantum state $(\ket{\mathcal{M}_{Y \rightarrow X_1 \rightarrow X_2}}+\ket{\mathcal{M}_{X_1 \rightarrow Y \rightarrow X_2}}+\ket{\mathcal{M}_{X_1  \rightarrow X_2\rightarrow Y}})/\sqrt{3}$ and can implement quantum operations (including measurements) on the mass state.

For Alice to Bob teleportation, they follow the following protocol -- If Alice receives a signal from Bob before her local time $t_A=\tau_1$ (i.e., event $X_1$), then she will implement unitary operator $U_{X_1|Y} = U_2$ and then unitary operator $U_{X_2} = U_1$ at time $t_A=\tau_2$ on her qutrit subsystem. However, if she does not receive signal from Bob before $t_A=\tau_1$, then she will implement $U_{X_1} = U_1$ on her qutrit subsystem and send a signal to Bob carrying $1$ bit message "$0$". Then she will wait for $t_A=\tau_2$, before which if she receives signal from Bob, she will implement $U_{X_2|Y} = U_3$ on her subsystem; otherwise, she will implement $U_{X_2} = U_1$ on her subsystem and send a signal to Bob carrying $1$ bit message "$1$". On the other hand, if Bob receives signal from Alice before $t_B=\tau_1$ carrying message "$0$", he will implement $U_{Y|X_1} = U_2$ on his ancillary qutrit system; if he receives the message "$1$" before $t_B=\tau_1$, he will implement $U_{Y|X_2} = U_3$ on his ancillary system. Whereas, if he does not receive any signal before $t_B=\tau_1$, he will implement $U_Y=U_1$. Bob, after implementation of any unitary operation, always sends signal to Alice. Here, 
\begin{align}
    U_1 = \mathbb{I} = \left(\begin{array}{ccc} 1  & 0 & 0 \\ 0 & 1 & 0 \\ 0 & 0 & 1\end{array}\right),~~ U_2 = \left(\begin{array}{ccc} 0  & 0 & 1 \\ 1 & 0 & 0 \\ 0 & 1 & 0\end{array}\right),~~ U_3 = \left(\begin{array}{ccc} 0  & 1 & 0 \\ 0 & 0 & 1 \\ 1 & 0 & 0\end{array}\right).
\end{align}

Thus, the joint state of Alice and Bob's shared system, Bob's ancilla and the mass configuration controlled by Charlie can be written as
\begin{align}
\frac{1}{\sqrt{3}} &\left(\ket{\mathcal{M}_{X_1  \rightarrow X_2\rightarrow Y}} (U_{X_2}\otimes\mathbb{I})(U_{X_1}\otimes\mathbb{I})\ket{\psi}_{AB} U_{Y|X_2}\ket{0}_{B'} + \ket{\mathcal{M}_{X_1 \rightarrow Y \rightarrow X_2}} (U_{X_2|Y}\otimes\mathbb{I})(U_{X_1}\otimes\mathbb{I})\ket{\psi}_{AB} U_{Y|X_1}\ket{0}_{B'}\right. \nonumber\\
&\left. +\ket{\mathcal{M}_{Y \rightarrow X_1 \rightarrow X_2}} (U_{X_2}\otimes\mathbb{I})(U_{X_1|Y}\otimes\mathbb{I})\ket{\psi}_{AB} U_Y\ket{0}_{B'}\right) \nonumber \\
=\frac{1}{\sqrt{3}} &\left\{\ket{\mathcal{M}_{X_1  \rightarrow X_2\rightarrow Y}} \sum_{j=0}^{d-1}\left(\alpha_{0j}\ket{0}_A\ket{j}_B+\alpha_{1j}\ket{1}_A\ket{j}_B+\alpha_{2j}\ket{2}_A\ket{j}_B\right) \ket{2}_{B'}\right.\nonumber\\
&\left.+ \ket{\mathcal{M}_{X_1 \rightarrow Y \rightarrow X_2}} \sum_{j=0}^{d-1}\left(\alpha_{0j}\ket{2}_A\ket{j}_B+\alpha_{1j}\ket{0}_A\ket{j}_B+\alpha_{2j}\ket{1}_A\ket{j}_B\right) \ket{1}_{B'}\right. \nonumber\\
&\left.+ \ket{\mathcal{M}_{Y \rightarrow X_1 \rightarrow X_2}} \sum_{j=0}^{d-1}\left(\alpha_{0j}\ket{1}_A\ket{j}_B+\alpha_{1j}\ket{2}_A\ket{j}_B+\alpha_{2j}\ket{0}_A\ket{j}_B\right) \ket{0}_{B'}\right\} \nonumber \\
=\frac{1}{3}&\left[\ket{a}\sum_{j=0}^{d-1}\left\{\left(\alpha_{0j}\ket{2}_{B'}\ket{j}_B+\alpha_{1j}\ket{1}_{B'}\ket{j}_B+\alpha_{2j}\ket{0}_{B'}\ket{j}_B\right)\ket{0}_A\right.\right. \nonumber\\
&~~~~~~~~~~~~~~+\left.\left.\left(\alpha_{1j}\ket{2}_{B'}\ket{j}_B+\alpha_{2j}\ket{1}_{B'}\ket{j}_B+\alpha_{0j}\ket{0}_{B'}\ket{j}_B\right)\ket{1}_A\right.\right.\nonumber\\
&~~~~~~~~~~~~~~+\left.\left.\left(\alpha_{2j}\ket{2}_{B'}\ket{j}_B+\alpha_{0j}\ket{1}_{B'}\ket{j}_B+\alpha_{1j}\ket{0}_{B'}\ket{j}_B\right)\ket{2}_A\rbrace\right.\right. \nonumber\\
+&\left.\ket{b}\sum_{j=0}^{d-1}\left\{\left(\alpha_{0j}\ket{2}_{B'}\ket{j}_B+\omega^2\alpha_{1j}\ket{1}_{B'}\ket{j}_B+\omega\alpha_{2j}\ket{0}_{B'}\ket{j}_B\right)\ket{0}_A\right.\right. \nonumber\\
&~~~~~~~~~~~~+\left.\left.\left(\alpha_{1j}\ket{2}_{B'}\ket{j}_B+\omega^2\alpha_{2j}\ket{1}_{B'}\ket{j}_B+\omega\alpha_{0j}\ket{0}_{B'}\ket{j}_B\right)\ket{1}_A\right.\right.\nonumber\\
&~~~~~~~~~~~~+\left.\left.\left(\alpha_{2j}\ket{2}_{B'}\ket{j}_B+\omega^2\alpha_{0j}\ket{1}_{B'}\ket{j}_B+\omega\alpha_{1j}\ket{0}_{B'}\ket{j}_B\right)\ket{2}_A\right\}\right. \nonumber\\
+&\left.\ket{c}\sum_{j=0}^{d-1}\left\{\left(\alpha_{0j}\ket{2}_{B'}\ket{j}_B+\omega\alpha_{1j}\ket{1}_{B'}\ket{j}_B+\omega^2\alpha_{2j}\ket{0}_{B'}\ket{j}_B\right)\ket{0}_A\right.\right. \nonumber\\
&~~~~~~~~~~~~+\left.\left.\left(\alpha_{1j}\ket{2}_{B'}\ket{j}_B+\omega\alpha_{2j}\ket{1}_{B'}\ket{j}_B+\omega^2\alpha_{0j}\ket{0}_{B'}\ket{j}_B\right)\ket{1}_A\right.\right.\nonumber\\
&~~~~~~~~~~~~+\left.\left.\left(\alpha_{2j}\ket{2}_{B'}\ket{j}_B+\omega\alpha_{0j}\ket{1}_{B'}\ket{j}_B+\omega^2\alpha_{1j}\ket{0}_{B'}\ket{j}_B\right)\ket{2}_A\right\}\right],
\end{align}
where
\begin{align}
    \ket{a} &= \frac{1}{\sqrt{3}}(\ket{\mathcal{M}_{X_1  \rightarrow X_2\rightarrow Y}} + \ket{\mathcal{M}_{X_1 \rightarrow Y \rightarrow X_2}} + \ket{\mathcal{M}_{Y \rightarrow X_1 \rightarrow X_2}}),\nonumber \\
    \ket{b} &= \frac{1}{\sqrt{3}}(\ket{\mathcal{M}_{X_1  \rightarrow X_2\rightarrow Y}} + \omega \ket{\mathcal{M}_{X_1 \rightarrow Y \rightarrow X_2}} + \omega^2 \ket{\mathcal{M}_{Y \rightarrow X_1 \rightarrow X_2}}),\nonumber \\
    \ket{c} &= \frac{1}{\sqrt{3}}(\ket{\mathcal{M}_{X_1  \rightarrow X_2\rightarrow Y}} + \omega^2\ket{\mathcal{M}_{X_1 \rightarrow Y \rightarrow X_2}} + \omega\ket{\mathcal{M}_{Y \rightarrow X_1 \rightarrow X_2}})\nonumber \\
\end{align}
with $\omega = e^{i\frac{2\pi}{3}}$ being the cube root of unity. Clearly, these three states- $\ket{a}$, $\ket{b}$, $\ket{c}$  are orthonormal.

Now, Charlie measures the mass state in $\lbrace \ket{a},\ket{b},\ket{c} \rbrace$ basis and classically communicate his outcome $\lbrace a,b,c\rbrace$ to Bob.  
Also, Alice measures her qutrit in $\lbrace |0\rangle_A ,|1\rangle_A , \ket{2}_A \rbrace$ basis and communicate her outcome $i \in \lbrace 0,1,2\rbrace$ classically to Bob. Bob, upon receiving messages from Alice and Charlie, will accordingly make a correction operation on his ancillary system to get $\sum_{i=0}^2\sum_{j=0}^{d-1}\alpha_{ij}\ket{i}_{B'}\ket{j}_B$ as the joint state of ($B B'$). These are mentioned in the Table \ref{tab3}. Bob's correction operations mentioned in this table can be determined from the following expressions, 
\begin{align}
    &V_1 = \left(\begin{array}{ccc} 0  & 0 & 1 \\ 0 & 1 & 0 \\ 1 & 0 & 0\end{array}\right),~~ V_2 = \left(\begin{array}{ccc} 1  & 0 & 0 \\ 0 & 0 & 1 \\ 0 & 1 & 0\end{array}\right),~~ V_3 = \left(\begin{array}{ccc} 0  & 1 & 0 \\ 1 & 0 & 0 \\ 0 & 0 & 1\end{array}\right), ~~ \Omega_1 = \left(\begin{array}{ccc} \omega^2  & 0 & 0 \\ 0 & \omega & 0 \\ 0 & 0 & 1\end{array}\right),~~ \Omega_2 = \left(\begin{array}{ccc} \omega  & 0 & 0 \\ 0 & \omega^2 & 0 \\ 0 & 0 & 1\end{array}\right).
    \label{bobcorr3ics}
\end{align}
This completes Alice to Bob teleportation process. After the teleportation, Bob has $\ket{\psi}_{B'B}$ is his local laboratory, Alice has either $\ket{0}_A$ or $\ket{1}_A$ or $\ket{2}_A$, and Charlie has mass configuration in either $\ket{a}$ or $\ket{b}$ or $\ket{c}$.

\begin{table*}[t]
   { \centering
\begin{tabular}{||c|c|c|c|c||} 
 \hline
 \begin{tabular}[c]{@{}l@{}}Charlie's\\ outcome\end{tabular} & \begin{tabular}[c]{@{}l@{}}Alice's\\ outcome\end{tabular} & Bob's collapsed state & \begin{tabular}[c]{@{}l@{}}Bob's correction\\ operation\end{tabular} & Bob's final state \\ [0.5ex]
 \hline\hline
 $a$ & $0$ & $\sum_{j=0}^{d-1}\left(\alpha_{0j}\ket{2}_{B'}\ket{j}_B+\alpha_{1j}\ket{1}_{B'}\ket{j}_B+\alpha_{2j}\ket{0}_{B'}\ket{j}_B\right)$ & $V_1\otimes\mathbb{I}$ & \multirow{9}{*}{$\sum_{i=0}^2\sum_{j=0}^{d-1}\alpha_{ij}\ket{i}_{B'}\ket{j}_B$}\\ 
 \cline{1-4}
 $a$ & $1$ & $\sum_{j=0}^{d-1}\left(\alpha_{1j}\ket{2}_{B'}\ket{j}_B+\alpha_{2j}\ket{1}_{B'}\ket{j}_B+\alpha_{0j}\ket{0}_{B'}\ket{j}_B\right)$ & $V_2\otimes\mathbb{I}$ & \\
  \cline{1-4}
 $a$ & $2$ & $\sum_{j=0}^{d-1}\left(\alpha_{2j}\ket{2}_{B'}\ket{j}_B+\alpha_{0j}\ket{1}_{B'}\ket{j}_B+\alpha_{1j}\ket{0}_{B'}\ket{j}_B\right)$ & $V_3\otimes\mathbb{I}$ & \\
 \cline{1-4} 
 $b$ & $0$ & $\sum_{j=0}^{d-1}\left(\alpha_{0j}\ket{2}_{B'}\ket{j}_B+\omega^2\alpha_{1j}\ket{1}_{B'}\ket{j}_B+\omega\alpha_{2j}\ket{0}_{B'}\ket{j}_B\right)$ & $ V_1\Omega_1\otimes\mathbb{I}$ &\\
  \cline{1-4}
 $b$ & $1$ & $\sum_{j=0}^{d-1}\left(\alpha_{1j}\ket{2}_{B'}\ket{j}_B+\omega^2\alpha_{2j}\ket{1}_{B'}\ket{j}_B+\omega\alpha_{0j}\ket{0}_{B'}\ket{j}_B\right)$ & $ V_2\Omega_1\otimes\mathbb{I}$ &\\
  \cline{1-4}
 $b$ & $2$ & $\sum_{j=0}^{d-1}\left(\alpha_{2j}\ket{2}_{B'}\ket{j}_B+\omega^2\alpha_{0j}\ket{1}_{B'}\ket{j}_B+\omega\alpha_{1j}\ket{0}_{B'}\ket{j}_B\right)$ & $ V_3\Omega_1\otimes\mathbb{I}$ &\\
  \cline{1-4}
 $c$ & $0$ & $\sum_{j=0}^{d-1}\left(\alpha_{0j}\ket{2}_{B'}\ket{j}_B+\omega\alpha_{1j}\ket{1}_{B'}\ket{j}_B+\omega^2\alpha_{2j}\ket{0}_{B'}\ket{j}_B\right)$ & $V_1\Omega_2\otimes\mathbb{I}$ &\\
 \cline{1-4}
 $c$ & $1$ & $\sum_{j=0}^{d-1}\left(\alpha_{1j}\ket{2}_{B'}\ket{j}_B+\omega\alpha_{2j}\ket{1}_{B'}\ket{j}_B+\omega^2\alpha_{0j}\ket{0}_{B'}\ket{j}_B\right)$ & $V_2\Omega_2\otimes\mathbb{I}$ &\\
 \cline{1-4}
 $c$ & $2$ & $\sum_{j=0}^{d-1}\left(\alpha_{2j}\ket{2}_{B'}\ket{j}_B+\omega\alpha_{0j}\ket{1}_{B'}\ket{j}_B+\omega^2\alpha_{1j}\ket{0}_{B'}\ket{j}_B\right)$ & $V_3\Omega_2\otimes\mathbb{I}$ &\\
 \hline
\end{tabular}
\caption{Details of the teleportation protocol from Alice to Bob using $3$-ICS for locally implementing any quantum operation on an arbitrary pure quantum state $|\psi_{AB} \rangle \in  \mathbb{C}^3\otimes\mathbb{C}^d$ shared by Alice and Bob. Bob's correction operations mentioned here can be evaluated from Eq.(\ref{bobcorr3ics}).}
\label{tab3}}
\end{table*}

We will now describe the back-teleportation protocol from Bob to Alice. Let us consider an arbitrary state $\ket{\psi}_{B'B} = \sum_{i=0}^2 \sum_{j=0}^{d-1} \alpha_{ij} \ket{i}_{B'}\ket{j}_B \in \mathbb{C}^3\otimes\mathbb{C}^d$ (with $\sum_{i,j} |\alpha_{ij}|^2=1$) of Bob's ancillary qutrit and his $d$-dimensional subsystem. To start with, Alice resets her qutrit to $\ket{0}_A$, and Charlie resets his mass configuration to $\ket{a}$. Then, Alice and Bob follow the same protocol to implement respective unitary operations on Alice's system and Bob's ancillary system, just as earlier. Now, the joint state of Alice's qutrit, Bob's systems and the mass configuration controlled by Charlie can be written as

\begin{align}
    \frac{1}{\sqrt{3}} &\left(\ket{\mathcal{M}_{X_1  \rightarrow X_2\rightarrow Y}} U_{X_2}U_{X_1}\ket{0}_{A} (U_{Y|X_2}\otimes\mathbb{I})\ket{\psi}_{B'B} + \ket{\mathcal{M}_{X_1 \rightarrow Y \rightarrow X_2}} U_{X_2|Y}U_{X_1}\ket{0}_{A} (U_{Y|X_1}\otimes\mathbb{I})\ket{\psi}_{B'B}\right. \nonumber\\
    &\left. +\ket{\mathcal{M}_{Y \rightarrow X_1 \rightarrow X_2}} U_{X_2}U_{X_1|Y}\ket{0}_{A} (U_Y\otimes\mathbb{I})\ket{\psi}_{B'B}\right) \nonumber \\
    =\frac{1}{\sqrt{3}} &\left\{\ket{\mathcal{M}_{X_1  \rightarrow X_2\rightarrow Y}}\ket{0}_A \sum_{j=0}^{d-1}\left(\alpha_{0j}\ket{2}_{B'}\ket{j}_B+\alpha_{1j}\ket{0}_{B'}\ket{j}_B+\alpha_{2j}\ket{1}_{B'}\ket{j}_B\right)\right.\nonumber\\
    &\left.+ \ket{\mathcal{M}_{X_1 \rightarrow Y \rightarrow X_2}}\ket{2}_A \sum_{j=0}^{d-1}\left(\alpha_{0j}\ket{1}_{B'}\ket{j}_B+\alpha_{1j}\ket{2}_{B'}\ket{j}_B+\alpha_{2j}\ket{0}_{B'}\ket{j}_B\right)\right. \nonumber\\
    &\left.+ \ket{\mathcal{M}_{Y \rightarrow X_1 \rightarrow X_2}}\ket{1}_A \sum_{j=0}^{d-1}\left(\alpha_{0j}\ket{0}_{B'}\ket{j}_B+\alpha_{1j}\ket{1}_{B'}\ket{j}_B+\alpha_{2j}\ket{2}_{B'}\ket{j}_B\right)\right\} \nonumber \\
    =\frac{1}{3}&\left[\ket{a}\sum_{j=0}^{d-1}\left\{\left(\alpha_{1j}\ket{0}_{A}\ket{j}_B+\alpha_{2j}\ket{2}_{A}\ket{j}_B+\alpha_{0j}\ket{1}_{A}\ket{j}_B\right)\ket{0}_{B'}\right.\right.\nonumber\\
    &~~~~~~~~~~~~ +\left.\left.\left(\alpha_{2j}\ket{0}_{A}\ket{j}_B+\alpha_{0j}\ket{2}_A\ket{j}_B+\alpha_{1j}\ket{1}_{A}\ket{j}_B\right)\ket{1}_{B'}\right.\right.\nonumber\\
    &~~~~~~~~~~~~~~+\left.\left.\left(\alpha_{0j}\ket{0}_{A}\ket{j}_B+\alpha_{1j}\ket{2}_{A}\ket{j}_B+\alpha_{2j}\ket{1}_{A}\ket{j}_B\right)\ket{2}_{B'}\rbrace\right.\right. \nonumber\\
    +&\left.\ket{b}\sum_{j=0}^{d-1}\left\{\left(\alpha_{1j}\ket{0}_A\ket{j}_B+\omega^2\alpha_{2j}\ket{2}_A\ket{j}_B+\omega\alpha_{0j}\ket{1}_A\ket{j}_B\right)\ket{0}_{B'}\right.\right. \nonumber\\
    &~~~~~~~~~~~~+\left.\left.\left(\alpha_{2j}\ket{0}_A\ket{j}_B+\omega^2\alpha_{0j}\ket{2}_A\ket{j}_B+\omega\alpha_{1j}\ket{1}_A\ket{j}_B\right)\ket{1}_{B'}\right.\right.\nonumber\\
    &~~~~~~~~~~~~+\left.\left.\left(\alpha_{0j}\ket{0}_{A}\ket{j}_B+\omega^2\alpha_{1j}\ket{2}_{A}\ket{j}_B+\omega\alpha_{2j}\ket{1}_{A}\ket{j}_B\right)\ket{2}_{B'}\right\}\right. \nonumber\\
    +&\left.\ket{c}\sum_{j=0}^{d-1}\left\{\left(\alpha_{1j}\ket{0}_{A}\ket{j}_B+\omega\alpha_{2j}\ket{2}_{A}\ket{j}_B+\omega^2\alpha_{0j}\ket{1}_{A}\ket{j}_B\right)\ket{0}_{B'}\right.\right. \nonumber\\
    &~~~~~~~~~~~~+\left.\left.\left(\alpha_{2j}\ket{0}_{A}\ket{j}_B+\omega\alpha_{0j}\ket{2}_{A}\ket{j}_B+\omega^2\alpha_{1j}\ket{1}_{A}\ket{j}_B\right)\ket{1}_{B'}\right.\right.\nonumber\\
    &~~~~~~~~~~~~+\left.\left.\left(\alpha_{0j}\ket{0}_{A}\ket{j}_B+\omega\alpha_{1j}\ket{2}_{A}\ket{j}_B+\omega^2\alpha_{2j}\ket{1}_{A}\ket{j}_B\right)\ket{2}_{B'}\right\}\right].
\end{align}

Now, Charlie measures the mass state in $\lbrace \ket{a},\ket{b},\ket{c} \rbrace$ basis and classically communicate his outcome $\lbrace a,b,c\rbrace$ to Alice; and Bob measures his ancillary system in $\lbrace |0\rangle_{B'} ,|1\rangle_{B'} , \ket{2}_{B'} \rbrace$ basis and communicate his outcome $i \in \lbrace 0,1,2\rbrace$ classically to Alice. Alice, upon receiving messages from Bob and Charlie, will accordingly make a correction operation on her subsystem such that the $3 \times d$-dimensional shared state between Alice and Bob becomes $\sum_{i=0}^2\sum_{j=0}^{d-1}\alpha_{ij}\ket{i}_{A}\ket{j}_B$. These details are mentioned in the Table \ref{tab4}.

\begin{table*}[t]
   { \centering
\begin{tabular}{||c|c|c|c|c||} 
 \hline
 \begin{tabular}[c]{@{}l@{}}Charlie's\\ outcome\end{tabular} & \begin{tabular}[c]{@{}l@{}}Bob's\\ outcome\end{tabular} & Alice and Bob's collapsed state & \begin{tabular}[c]{@{}l@{}}Alice's correction\\ operation\end{tabular} & \begin{tabular}[c]{@{}l@{}}Alice and Bob's\\ final state\end{tabular} \\ [0.5ex]
 \hline\hline
 $a$ & $0$ & $\sum_{j=0}^{d-1}\left(\alpha_{1j}\ket{0}_{A}\ket{j}_B+\alpha_{2j}\ket{2}_{A}\ket{j}_B+\alpha_{0j}\ket{1}_{A}\ket{j}_B\right)$ & $V_3$ & \multirow{9}{*}{$\sum_{i=0}^2\sum_{j=0}^{d-1}\alpha_{ij}\ket{i}_{A}\ket{j}_B$}\\ 
 \cline{1-4}
 $a$ & $1$ & $\sum_{j=0}^{d-1}\left(\alpha_{2j}\ket{0}_{A}\ket{j}_B+\alpha_{0j}\ket{2}_A\ket{j}_B+\alpha_{1j}\ket{1}_{A}\ket{j}_B\right)$ & $V_1$ & \\
  \cline{1-4}
 $a$ & $2$ & $\sum_{j=0}^{d-1}\left(\alpha_{0j}\ket{0}_{A}\ket{j}_B+\alpha_{1j}\ket{2}_{A}\ket{j}_B+\alpha_{2j}\ket{1}_{A}\ket{j}_B\right)$ & $V_2$ & \\
 \cline{1-4} 
 $b$ & $0$ & $\sum_{j=0}^{d-1}\left(\alpha_{1j}\ket{0}_A\ket{j}_B+\omega^2\alpha_{2j}\ket{2}_A\ket{j}_B+\omega\alpha_{0j}\ket{1}_A\ket{j}_B\right)$ & $ V_3\Omega_1$ &\\
  \cline{1-4}
 $b$ & $1$ & $\sum_{j=0}^{d-1}\left(\alpha_{2j}\ket{0}_A\ket{j}_B+\omega^2\alpha_{0j}\ket{2}_A\ket{j}_B+\omega\alpha_{1j}\ket{1}_A\ket{j}_B\right)$ & $ V_1\Omega_1$ &\\
  \cline{1-4}
 $b$ & $2$ & $\sum_{j=0}^{d-1}\left(\alpha_{0j}\ket{0}_{A}\ket{j}_B+\omega^2\alpha_{1j}\ket{2}_{A}\ket{j}_B+\omega\alpha_{2j}\ket{1}_{A}\ket{j}_B\right)$ & $ V_2\Omega_1$ &\\
  \cline{1-4}
 $c$ & $0$ & $\sum_{j=0}^{d-1}\left(\alpha_{1j}\ket{0}_{A}\ket{j}_B+\omega\alpha_{2j}\ket{2}_{A}\ket{j}_B+\omega^2\alpha_{0j}\ket{1}_{A}\ket{j}_B\right)$ & $V_3\Omega_2 $ &\\
 \cline{1-4}
 $c$ & $1$ & $\sum_{j=0}^{d-1}\left(\alpha_{2j}\ket{0}_{A}\ket{j}_B+\omega\alpha_{0j}\ket{2}_{A}\ket{j}_B+\omega^2\alpha_{1j}\ket{1}_{A}\ket{j}_B\right)$ & $V_1\Omega_2 $ &\\
 \cline{1-4}
 $c$ & $2$ & $\sum_{j=0}^{d-1}\left(\alpha_{0j}\ket{0}_{A}\ket{j}_B+\omega\alpha_{1j}\ket{2}_{A}\ket{j}_B+\omega^2\alpha_{2j}\ket{1}_{A}\ket{j}_B\right)$ & $V_2\Omega_2 $ &\\
 \hline
\end{tabular}
\caption{Details of the back-teleportation protocol from Bob to Alice using $3$-ICS for locally implementing any quantum operation on an arbitrary pure quantum state $|\psi_{AB} \rangle \in  \mathbb{C}^3\otimes\mathbb{C}^d$ shared by Alice and Bob. Alice's correction operations mentioned here can be evaluated from Eq.(\ref{bobcorr3ics}).}
\label{tab4}}
\end{table*}

\subsubsection{$4$-ICS as universal resource for local implementation of nonlocal quantum operations on quantum states $\in \mathbb{C}^4\otimes\mathbb{C}^d$}

Consider Alice and Bob share an arbitrary quantum state $\ket{\psi}_{AB}=\sum_{i=0}^3\sum_{j=0}^{d-1}\alpha_{ij}\ket{i}_A\ket{j}_B \in \mathbb{C}^4\otimes\mathbb{C}^d$ (where $\sum_{i,j}|\alpha_{ij}|^2=1$).  Bob has an additional ancillary ququad system in $\ket{0}_{B'}$ state. Here also, throughout this section, Bob's subsystem of the initially shared state will be denoted by $B$ and the state of Bob's ancillary ququad will be denoted by $B'$. Let us define three events -- $X_1$, $X_2$, and $X_3$ in Alice's laboratory as Alice's local times $t_A=\tau_1$, $\tau_2$,  and $\tau_3$ with $\tau_1<\tau_2<\tau_3$, and one event $Y$ at Bob's laboratory as his local time $t_B=\tau_1$. Since, the events $X_1$, $X_2$, and $X_3$ are defined in the same location, they have fixed causal order, i.e., $X_1 \rightarrow X_2 \rightarrow X_3$. Now, depending on the geometry of the classical space-time, the events can have four classically distinct and mutually exclusive causal orders --\\
\\
(1) event $Y$ is in causal past of $X_1$ (which we denote $Y \rightarrow X_1 \rightarrow X_2 \rightarrow X_3$),\\
(2) event $Y$ is in causal future of $X_1$, but in causal past of $X_2$ and $X_3$ (which we denote $X_1 \rightarrow Y \rightarrow X_2 \rightarrow X_3$),\\
(3) event $Y$ is in causal future of $X_1$ and $X_2$, but in causal past of $X_3$ (which we denote $X_1  \rightarrow X_2 \rightarrow Y \rightarrow X_3$), and\\
(4) event $Y$ is in causal future of $X_1$, $X_2$ and $X_3$ (which we denote $X_1  \rightarrow X_2 \rightarrow X_3 \rightarrow Y$).\\
\\

Let us denote the corresponding quantum states of mass configurations, each yielding one of the four space-time geometries with the above-mentioned causal orders, as $\ket{\mathcal{M}_{Y \rightarrow X_1 \rightarrow X_2 \rightarrow X_3}}$, $\ket{\mathcal{M}_{X_1 \rightarrow Y \rightarrow X_2 \rightarrow X_3}}$, $\ket{\mathcal{M}_{X_1  \rightarrow X_2 \rightarrow Y \rightarrow X_3}}$ and $\ket{\mathcal{M}_{X_1  \rightarrow X_2 \rightarrow X_3 \rightarrow Y}}$. We consider that Charlie can prepare the mass configuration in a coherent superposition of these states, i.e., in the state $(\ket{\mathcal{M}_{X_1  \rightarrow X_2 \rightarrow X_3 \rightarrow Y}} + \ket{\mathcal{M}_{X_1  \rightarrow X_2 \rightarrow Y \rightarrow X_3}} + \ket{\mathcal{M}_{X_1 \rightarrow Y \rightarrow X_2 \rightarrow X_3}} + \ket{\mathcal{M}_{Y \rightarrow X_1 \rightarrow X_2 \rightarrow X_3}})/2$ and can implement quantum operations (including measurements) on the mass state.

For Alice to Bob teleportation, they follow the following protocol -- If Alice receives a signal from Bob before her local time $t_A=\tau_1$ (i.e., event $X_1$), then she will implement unitary operator $U_{X_1|Y} = U_2$ and then unitary operator $U_{X_2} = U_1, U_{X_3} = U_2$ at times $t_A=\tau_2 \text{ and }\tau_3$ respectively, on her ququad subsystem. However, if she does not receive signal from Bob before $t_A=\tau_1$, then she will implement $U_{X_1} = U_1$ on her ququad subsystem and send a signal to Bob carrying message "$0$". Then she will wait for $t_A=\tau_2$, before which if she receives signal from Bob, she will implement $U_{X_2|Y} = U_3$ at $t_A=\tau_2$, followed by $U_{X_3} = U_1$ at $t_A=\tau_3$ on her subsystem; otherwise, she will implement $U_{X_2} = U_1$ on her subsystem and send a signal to Bob carrying message "$1$". If she receives further signal from Bob before $t_A=\tau_3$, she will implement $U_{X_3|Y} = U_4$ at $t_A=\tau_3$, else she will implement $U_{X_3} = U_1$ and send a signal to Bob carrying message "$2$". On the other hand, if Bob receives signal from Alice before $t_B=\tau_1$, he will implement $U_{Y|X_1} = U_2,U_{Y|X_2} = U_3$ and $U_{Y|X_3} = U_4$ on his ancillary ququad system depending on whether he receives "$0$","$1$", or "$2$". Whereas, if he does not receive any signal before $t_B=\tau_1$, he will implement $U_Y=U_1$. Bob, after implementation of any unitary operation, always sends signal to Alice. Here,
\begin{align}
    U_1 = \left(\begin{array}{cccc} 1  & 0 & 0 & 0\\ 0 & 1 & 0 & 0\\ 0 & 0 & 1 &0 \\ 0 & 0 & 0 &1\end{array}\right)= \mathbb{I},~~ U_2 = \left(\begin{array}{cccc} 0  & 0 & 1 & 0\\ 0 & 0 & 0 & 1\\ 1 & 0 & 0 &0 \\ 0 & 1 & 0 & 0\end{array}\right),~~ U_3 = \left(\begin{array}{cccc} 0  & 0 & 0 & 1\\ 1 & 0 & 0 & 0\\ 0 & 1 & 0 &0 \\ 0 & 0 & 1 &0\end{array}\right),~~ U_4 = \left(\begin{array}{cccc} 0  & 1 & 0 & 0\\ 0 & 0 & 1 & 0\\ 0 & 0 & 0 &1 \\ 1 & 0 & 0 &0\end{array}\right).
\end{align}

\begin{table*}[t]
   { \centering
\begin{tabular}{||c|c|c|c|c||} 
 \hline
 \begin{tabular}[c]{@{}l@{}}Charlie's\\ outcome\end{tabular} & \begin{tabular}[c]{@{}l@{}}Alice's\\ outcome\end{tabular} & Bob's collapsed state & \begin{tabular}[c]{@{}l@{}}Bob's correction\\ operation\end{tabular} & \begin{tabular}[c]{@{}l@{}}Bob's\\ final state\end{tabular} \\ [0.5ex]
 \hline\hline
 $a$ & $0$ & $\sum_{j=0}^{d-1}(\alpha_{0j}\ket{3}_{B'}\ket{j}_B+\alpha_{1j}\ket{1}_{B'}\ket{j}_B+\alpha_{3j}\ket{2}_{B'}\ket{j}_B+\alpha_{2j}\ket{0}_{B'}\ket{j}_B)$ & $V_1\otimes\mathbb{I}$ & \multirow{16}{*}{$\sum_{i=0}^3\sum_{j=0}^{d-1}\alpha_{ij}\ket{i}_{B'}\ket{j}_{B}$}\\ 
 \cline{1-4}
 $a$ & $1$ & $\sum_{j=0}^{d-1}(\alpha_{1j}\ket{3}_{B'}\ket{j}_B+\alpha_{2j}\ket{1}_{B'}\ket{j}_B+\alpha_{0j}\ket{2}_{B'}\ket{j}_B+\alpha_{3j}\ket{0}_{B'}\ket{j}_B)$ & $V_2\otimes\mathbb{I}$ & \\
 \cline{1-4}
 $a$ & $2$ & $\sum_{j=0}^{d-1}(\alpha_{2j}\ket{3}_{B'}\ket{j}_B+\alpha_{3j}\ket{1}_{B'}\ket{j}_B+\alpha_{1j}\ket{2}_{B'}\ket{j}_B+\alpha_{0j}\ket{0}_{B'}\ket{j}_B)$ & $V_3\otimes\mathbb{I}$ & \\
 \cline{1-4}
 $a$ & $3$ & $\sum_{j=0}^{d-1}(\alpha_{3j}\ket{3}_{B'}\ket{j}_B+\alpha_{0j}\ket{1}_{B'}\ket{j}_B+\alpha_{2j}\ket{2}_{B'}\ket{j}_B+\alpha_{1j}\ket{0}_{B'}\ket{j}_B)$ & $V_4\otimes\mathbb{I}$ & \\
 \cline{1-4}
 $b$ & $0$ & $\sum_{j=0}^{d-1}(\alpha_{0j}\ket{3}_{B'}\ket{j}_B-\alpha_{1j}\ket{1}_{B'}\ket{j}_B-\alpha_{3j}\ket{2}_{B'}\ket{j}_B+\alpha_{2j}\ket{0}_{B'}\ket{j}_B)$ & $V_1\Omega_1\otimes\mathbb{I} $ &\\
 \cline{1-4}
 $b$ & $1$ & $\sum_{j=0}^{d-1}(\alpha_{1j}\ket{3}_{B'}\ket{j}_B-\alpha_{2j}\ket{1}_{B'}\ket{j}_B-\alpha_{0j}\ket{2}_{B'}\ket{j}_B+\alpha_{3j}\ket{0}_{B'}\ket{j}_B)$ & $V_2\Omega_1\otimes\mathbb{I} $ &\\
 \cline{1-4}
 $b$ & $2$ & $\sum_{j=0}^{d-1}(\alpha_{2j}\ket{3}_{B'}\ket{j}_B-\alpha_{3j}\ket{1}_{B'}\ket{j}_B-\alpha_{1j}\ket{2}_{B'}\ket{j}_B+\alpha_{0j}\ket{0}_{B'}\ket{j}_B)$ & $V_3\Omega_1\otimes\mathbb{I} $ &\\
 \cline{1-4}
 $b$ & $3$ & $\sum_{j=0}^{d-1}(\alpha_{3j}\ket{3}_{B'}\ket{j}_B-\alpha_{0j}\ket{1}_{B'}\ket{j}_B-\alpha_{2j}\ket{2}_{B'}\ket{j}_B+\alpha_{1j}\ket{0}_{B'}\ket{j}_B)$ & $V_4\Omega_1\otimes\mathbb{I} $ &\\
 \cline{1-4}
 $c$ & $0$ & $\sum_{j=0}^{d-1}(\alpha_{0j}\ket{3}_{B'}\ket{j}_B-\alpha_{1j}\ket{1}_{B'}\ket{j}_B+\alpha_{3j}\ket{2}_{B'}\ket{j}_B-\alpha_{2j}\ket{0}_{B'}\ket{j}_B)$ & $V_1\Omega_2\otimes\mathbb{I} $ &\\
 \cline{1-4}
 $c$ & $1$ & $\sum_{j=0}^{d-1}(\alpha_{1j}\ket{3}_{B'}\ket{j}_B-\alpha_{2j}\ket{1}_{B'}\ket{j}_B+\alpha_{0j}\ket{2}_{B'}\ket{j}_B-\alpha_{3j}\ket{0}_{B'}\ket{j}_B)$ & $V_2\Omega_2\otimes\mathbb{I} $ &\\
 \cline{1-4}
 $c$ & $2$ & $\sum_{j=0}^{d-1}(\alpha_{2j}\ket{3}_{B'}\ket{j}_B-\alpha_{3j}\ket{1}_{B'}\ket{j}_B+\alpha_{1j}\ket{2}_{B'}\ket{j}_B-\alpha_{0j}\ket{0}_{B'}\ket{j}_B)$ & $V_3\Omega_2\otimes\mathbb{I} $ &\\
 \cline{1-4}
 $c$ & $3$ & $\sum_{j=0}^{d-1}(\alpha_{3j}\ket{3}_{B'}\ket{j}_B-\alpha_{0j}\ket{1}_{B'}\ket{j}_B+\alpha_{2j}\ket{2}_{B'}\ket{j}_B-\alpha_{1j}\ket{0}_{B'}\ket{j}_B)$ & $V_4\Omega_2\otimes\mathbb{I} $ &\\
 \cline{1-4}
 $d$ & $0$ & $\sum_{j=0}^{d-1}(\alpha_{0j}\ket{3}_{B'}\ket{j}_B+\alpha_{1j}\ket{1}_{B'}\ket{j}_B-\alpha_{3j}\ket{2}_{B'}\ket{j}_B-\alpha_{2j}\ket{0}_{B'}\ket{j}_B)$ & $V_1\Omega_3\otimes\mathbb{I} $ &\\
 \cline{1-4}
 $d$ & $1$ & $\sum_{j=0}^{d-1}(\alpha_{1j}\ket{3}_{B'}\ket{j}_B+\alpha_{2j}\ket{1}_{B'}\ket{j}_B-\alpha_{0j}\ket{2}_{B'}\ket{j}_B-\alpha_{3j}\ket{0}_{B'}\ket{j}_B)$ & $V_2\Omega_3\otimes\mathbb{I} $ &\\
 \cline{1-4}
 $d$ & $2$ & $\sum_{j=0}^{d-1}(\alpha_{2j}\ket{3}_{B'}\ket{j}_B+\alpha_{3j}\ket{1}_{B'}\ket{j}_B-\alpha_{1j}\ket{2}_{B'}\ket{j}_B-\alpha_{0j}\ket{0}_{B'}\ket{j}_B)$ & $V_3\Omega_3\otimes\mathbb{I} $ &\\
 \cline{1-4}
 $d$ & $3$ & $\sum_{j=0}^{d-1}(\alpha_{3j}\ket{3}_{B'}\ket{j}_B+\alpha_{0j}\ket{1}_{B'}\ket{j}_B-\alpha_{2j}\ket{2}_{B'}\ket{j}_B-\alpha_{1j}\ket{0}_{B'}\ket{j}_B)$ & $V_4\Omega_3\otimes\mathbb{I} $ &\\
 \hline
\end{tabular}
\caption{Details of the teleportation protocol from Alice to Bob using $4$-ICS for locally implementing any quantum operation on an arbitrary pure quantum state $|\psi_{AB} \rangle \in  \mathbb{C}^4\otimes\mathbb{C}^d$ shared by Alice and Bob. Bob's correction operations mentioned here can be evaluated from Eq.(\ref{bobcorr4ics}).}
\label{tab5}}
\end{table*}

Thus, the joint state of Alice and Bob's shared system, Bob's ancilla and the mass configuration controlled by Charlie can be written as
\begin{align}
    \frac{1}{2}&\left(\ket{\mathcal{M}_{X_1  \rightarrow X_2 \rightarrow X_3 \rightarrow Y}} (U_{X_3}\otimes\mathbb{I})(U_{X_2}\otimes\mathbb{I})(U_{X_1}\otimes\mathbb{I})\ket{\psi}_{AB} U_{Y|X_3}\ket{0}_{B'}\right.\nonumber\\
    &+ \left.\ket{\mathcal{M}_{X_1 \rightarrow X_2 \rightarrow Y \rightarrow X_3}} (U_{X_3|Y}\otimes\mathbb{I})(U_{X_2}\otimes\mathbb{I})(U_{X_1}\otimes\mathbb{I})\ket{\psi}_{AB} U_{Y|X_2}\ket{0}_{B'}\right. \nonumber \\
    &+ \left.\ket{\mathcal{M}_{X_1 \rightarrow Y \rightarrow X_2 \rightarrow X_3}} (U_{X_3}\otimes\mathbb{I})(U_{X_2|Y}\otimes\mathbb{I})(U_{X_1}\otimes\mathbb{I})\ket{\psi}_{AB} U_{Y|X_1}\ket{0}_{B'}\right. \nonumber\\
    &+ \left.\ket{\mathcal{M}_{Y \rightarrow X_1 \rightarrow X_2 \rightarrow X_3}}(U_{X_3}\otimes\mathbb{I})(U_{X_2}\otimes\mathbb{I})(U_{X_1|Y}\otimes\mathbb{I})\ket{\psi}_{AB} U_{Y}\ket{0}_{B'}\right) \nonumber \\
    =\frac{1}{2}&\left(\ket{\mathcal{M}_{X_1  \rightarrow X_2 \rightarrow X_3 \rightarrow Y}} \ket{\psi}_{AB}\ket{3}_{B'} + \ket{\mathcal{M}_{X_1 \rightarrow X_2 \rightarrow Y \rightarrow X_3}}\ket{\phi}_{AB}\ket{1}_{B'}+\ket{\mathcal{M}_{X_1 \rightarrow Y \rightarrow X_2 \rightarrow X_3}}\ket{\lambda}_{AB}\ket{2}_{B'}\right. \nonumber\\
    &+ \left.\ket{\mathcal{M}_{Y \rightarrow X_1 \rightarrow X_2 \rightarrow X_3}}\ket{\mu}_{AB}\ket{0}_{B'}\right) \nonumber \\
    = \frac{1}{4}&\left\{\ket{a}(\ket{\psi}_{AB} \ket{3}_{B'} + \ket{\phi}_{AB}\ket{1}_{B'} + \ket{\lambda}_{AB} \ket{2}_{B'} + \ket{\mu}_{AB} \ket{0}_{B'})\right. \nonumber\\
    &+\left.\ket{b}(\ket{\psi}_{AB} \ket{3}_{B'} - \ket{\phi}_{AB}\ket{1}_{B'} - \ket{\lambda}_{AB} \ket{2}_{B'} + \ket{\mu}_{AB} \ket{0}_{B'})\right. \nonumber \\
    &+ \left.\ket{c}(\ket{\psi}_{AB} \ket{3}_{B'} - \ket{\phi}_{AB}\ket{1}_{B'} + \ket{\lambda}_{AB} \ket{2}_{B'} - \ket{\mu}_{AB} \ket{0}_{B'} )\right.\nonumber\\
    &+\left. \ket{d}(\ket{\psi}_{AB} \ket{3}_{B'} + \ket{\phi}_{AB}\ket{1}_{B'} - \ket{\lambda}_{AB} \ket{2}_{B'} - \ket{\mu}_{AB} \ket{0}_{B'} )\right\},
\end{align}
where 
\begin{align}
    \ket{\phi}_{AB}=(U_4\otimes\mathbb{I})\ket{\psi}_{AB} = \sum_{j=0}^{d-1}\left(\alpha_{0j}\ket{3}_A\ket{j}_B+\alpha_{1j}\ket{0}_A\ket{j}_B+\alpha_{2j}\ket{1}_A\ket{j}_B+\alpha_{3j}\ket{2}_A\ket{j}_B\right),\nonumber\\
    \ket{\lambda}_{AB}=(U_3\otimes\mathbb{I})\ket{\psi}_{AB} = \sum_{j=0}^{d-1}\left(\alpha_{0j}\ket{1}_A\ket{j}_B+\alpha_{1j}\ket{2}_A\ket{j}_B+\alpha_{2j}\ket{3}_A\ket{j}_B+\alpha_{3j}\ket{0}_A\ket{j}_B\right),\nonumber\\
    \ket{\mu}_{AB}=(U_2\otimes\mathbb{I})\ket{\psi}_{AB} = \sum_{j=0}^{d-1}\left(\alpha_{0j}\ket{2}_A\ket{j}_B+\alpha_{1j}\ket{3}_A\ket{j}_B+\alpha_{2j}\ket{0}_A\ket{j}_B+\alpha_{3j}\ket{1}_A\ket{j}_B\right),
\end{align}
and
\begin{align}
    \ket{a} &= \frac{1}{2} (\ket{\mathcal{M}_{X_1  \rightarrow X_2 \rightarrow X_3 \rightarrow Y}} + \ket{\mathcal{M}_{X_1  \rightarrow X_2 \rightarrow Y \rightarrow X_3}} + \ket{\mathcal{M}_{X_1 \rightarrow Y \rightarrow X_2 \rightarrow X_3}} + \ket{\mathcal{M}_{Y \rightarrow X_1 \rightarrow X_2 \rightarrow X_3}}),\nonumber \\
    \ket{b} &= \frac{1}{2} (\ket{\mathcal{M}_{X_1  \rightarrow X_2 \rightarrow X_3 \rightarrow Y}} - \ket{\mathcal{M}_{X_1  \rightarrow X_2 \rightarrow Y \rightarrow X_3}} - \ket{\mathcal{M}_{X_1 \rightarrow Y \rightarrow X_2 \rightarrow X_3}} + \ket{\mathcal{M}_{Y \rightarrow X_1 \rightarrow X_2 \rightarrow X_3}}),\nonumber \\
    \ket{c} &= \frac{1}{2} (\ket{\mathcal{M}_{X_1  \rightarrow X_2 \rightarrow X_3 \rightarrow Y}} - \ket{\mathcal{M}_{X_1  \rightarrow X_2 \rightarrow Y \rightarrow X_3}} + \ket{\mathcal{M}_{X_1 \rightarrow Y \rightarrow X_2 \rightarrow X_3}} - \ket{\mathcal{M}_{Y \rightarrow X_1 \rightarrow X_2 \rightarrow X_3}}),\nonumber \\
    \ket{d} &=\frac{1}{2} (\ket{\mathcal{M}_{X_1  \rightarrow X_2 \rightarrow X_3 \rightarrow Y}} + \ket{\mathcal{M}_{X_1  \rightarrow X_2 \rightarrow Y \rightarrow X_3}} - \ket{\mathcal{M}_{X_1 \rightarrow Y \rightarrow X_2 \rightarrow X_3}} - \ket{\mathcal{M}_{Y \rightarrow X_1 \rightarrow X_2 \rightarrow X_3}}).
\end{align}

\begin{table*}[t]
   { \centering
\begin{tabular}{||c|c|c|c|c||} 
 \hline
 \begin{tabular}[c]{@{}l@{}}Charlie's\\ outcome\end{tabular} & \begin{tabular}[c]{@{}l@{}}Bob's\\ outcome\end{tabular} & Alice and Bob's collapsed state & \begin{tabular}[c]{@{}l@{}}Alice's correction\\ operation\end{tabular} & \begin{tabular}[c]{@{}l@{}}Alice and Bob's\\ final state\end{tabular} \\ [0.5ex]
 \hline\hline
 $a$ & $0$ & $\sum_{j=0}^{d-1}(\alpha_{1j}\ket{0}_{A}\ket{j}_B+\alpha_{3j}\ket{3}_{A}\ket{j}_B+\alpha_{2j}\ket{1}_{A}\ket{j}_B+\alpha_{0j}\ket{2}_{A}\ket{j}_B)$ & $W_1$ & \multirow{16}{*}{$\sum_{i=0}^3\sum_{j=0}^{d-1}\alpha_{ij}\ket{i}_{A}\ket{j}_{B}$}\\ 
 \cline{1-4}
 $a$ & $1$ & $\sum_{j=0}^{d-1}(\alpha_{2j}\ket{0}_{A}\ket{j}_B+\alpha_{0j}\ket{3}_{A}\ket{j}_B+\alpha_{3j}\ket{1}_{A}\ket{j}_B+\alpha_{1j}\ket{2}_{A}\ket{j}_B)$ & $W_2$ & \\
 \cline{1-4}
 $a$ & $2$ & $\sum_{j=0}^{d-1}(\alpha_{3j}\ket{0}_{A}\ket{j}_B+\alpha_{1j}\ket{3}_{A}\ket{j}_B+\alpha_{0j}\ket{1}_{A}\ket{j}_B+\alpha_{2j}\ket{2}_{A}\ket{j}_B)$ & $W_3$ & \\
 \cline{1-4}
 $a$ & $3$ & $\sum_{j=0}^{d-1}(\alpha_{0j}\ket{0}_{A}\ket{j}_B+\alpha_{2j}\ket{3}_{A}\ket{j}_B+\alpha_{1j}\ket{1}_{A}\ket{j}_B+\alpha_{3j}\ket{2}_{A}\ket{j}_B)$ & $W_4$ & \\
 \cline{1-4}
 $b$ & $0$ & $\sum_{j=0}^{d-1}(\alpha_{1j}\ket{0}_{A}\ket{j}_B-\alpha_{3j}\ket{3}_{A}\ket{j}_B-\alpha_{2j}\ket{1}_{A}\ket{j}_B+\alpha_{0j}\ket{2}_{A}\ket{j}_B)$ & $W_1\Omega_1$ &\\
 \cline{1-4}
 $b$ & $1$ & $\sum_{j=0}^{d-1}(\alpha_{2j}\ket{0}_{A}\ket{j}_B-\alpha_{0j}\ket{3}_{A}\ket{j}_B-\alpha_{3j}\ket{1}_{A}\ket{j}_B+\alpha_{1j}\ket{2}_{A}\ket{j}_B)$ & $W_2\Omega_1$ &\\
 \cline{1-4}
 $b$ & $2$ & $\sum_{j=0}^{d-1}(\alpha_{3j}\ket{0}_{A}\ket{j}_B-\alpha_{1j}\ket{3}_{A}\ket{j}_B-\alpha_{0j}\ket{1}_{A}\ket{j}_B+\alpha_{2j}\ket{2}_{A}\ket{j}_B)$ & $W_3\Omega_1$ &\\
 \cline{1-4}
 $b$ & $3$ & $\sum_{j=0}^{d-1}(\alpha_{0j}\ket{0}_{A}\ket{j}_B-\alpha_{2j}\ket{3}_{A}\ket{j}_B-\alpha_{1j}\ket{1}_{A}\ket{j}_B+\alpha_{3j}\ket{2}_{A}\ket{j}_B)$ & $W_4\Omega_1$ &\\
 \cline{1-4}
 $c$ & $0$ & $\sum_{j=0}^{d-1}(\alpha_{1j}\ket{0}_{A}\ket{j}_B-\alpha_{3j}\ket{3}_{A}\ket{j}_B+\alpha_{2j}\ket{1}_{A}\ket{j}_B-\alpha_{0j}\ket{2}_{A}\ket{j}_B)$ & $W_1\Omega_2$ &\\
 \cline{1-4}
 $c$ & $1$ & $\sum_{j=0}^{d-1}(\alpha_{2j}\ket{0}_{A}\ket{j}_B-\alpha_{0j}\ket{3}_{A}\ket{j}_B+\alpha_{3j}\ket{1}_{A}\ket{j}_B-\alpha_{1j}\ket{2}_{A}\ket{j}_B)$ & $W_2\Omega_2$ &\\
 \cline{1-4}
 $c$ & $2$ & $\sum_{j=0}^{d-1}(\alpha_{3j}\ket{0}_{A}\ket{j}_B-\alpha_{1j}\ket{3}_{A}\ket{j}_B+\alpha_{0j}\ket{1}_{A}\ket{j}_B-\alpha_{2j}\ket{2}_{A}\ket{j}_B)$ & $W_3\Omega_2$ &\\
 \cline{1-4}
 $c$ & $3$ & $\sum_{j=0}^{d-1}(\alpha_{0j}\ket{0}_{A}\ket{j}_B-\alpha_{2j}\ket{3}_{A}\ket{j}_B+\alpha_{1j}\ket{1}_{A}\ket{j}_B-\alpha_{3j}\ket{2}_{A}\ket{j}_B)$ & $W_4\Omega_2$ &\\
 \cline{1-4}
 $d$ & $0$ & $\sum_{j=0}^{d-1}(\alpha_{1j}\ket{0}_{A}\ket{j}_B+\alpha_{3j}\ket{3}_{A}\ket{j}_B-\alpha_{2j}\ket{1}_{A}\ket{j}_B-\alpha_{0j}\ket{2}_{A}\ket{j}_B)$ & $W_1\Omega_3$ &\\
 \cline{1-4}
 $d$ & $1$ & $\sum_{j=0}^{d-1}(\alpha_{2j}\ket{0}_{A}\ket{j}_B+\alpha_{0j}\ket{3}_{A}\ket{j}_B-\alpha_{3j}\ket{1}_{A}\ket{j}_B-\alpha_{1j}\ket{2}_{A}\ket{j}_B)$ & $W_2\Omega_3$ &\\
 \cline{1-4}
 $d$ & $2$ & $\sum_{j=0}^{d-1}(\alpha_{3j}\ket{0}_{A}\ket{j}_B+\alpha_{1j}\ket{3}_{A}\ket{j}_B-\alpha_{0j}\ket{1}_{A}\ket{j}_B-\alpha_{2j}\ket{2}_{A}\ket{j}_B)$ & $W_3\Omega_3$ &\\
 \cline{1-4}
 $d$ & $3$ & $\sum_{j=0}^{d-1}(\alpha_{0j}\ket{0}_{A}\ket{j}_B+\alpha_{2j}\ket{3}_{A}\ket{j}_B-\alpha_{1j}\ket{1}_{A}\ket{j}_B-\alpha_{3j}\ket{2}_{A}\ket{j}_B)$ & $W_4\Omega_3$ &\\
 \hline
\end{tabular}
\caption{Details of the back-teleportation protocol from Bob to Alice using $4$-ICS for locally implementing any quantum operation on an arbitrary pure quantum state $|\psi_{AB} \rangle \in  \mathbb{C}^4\otimes\mathbb{C}^d$ shared by Alice and Bob. Alice's correction operations mentioned here can be evaluated from Eq.(\ref{Alicecorr4ics}).}
\label{tab6}}
\end{table*}

Now Charlie measures the mass state in $\lbrace \ket{a},\ket{b},\ket{c},\ket{d} \rbrace$ basis and classically communicates his outcomes $\lbrace a,b,c,d \rbrace$ to Bob. Alice also measure her subsystem in $\lbrace |0\rangle_A ,|1\rangle_A , \ket{2}_A , \ket{3}_A \rbrace$ basis and communicates her outcomes $\lbrace 0,1,2,3 \rbrace$ classically to Bob. Bob then accordingly makes a correction operation on his ancillary system to get  $\ket{\psi}_{B'B} = \sum_{i=0}^2\sum_{j=0}^{d-1}\alpha_{ij}\ket{i}_{B'}\ket{j}_B$  as the joint state of ($B B'$). These details are presented in the Table \ref{tab5}. Bob's correction operations mentioned in this table can be determined from the following expressions, 
\begin{align}
    &V_1 = \left(\begin{array}{cccc} 0  & 0 & 0 & 1\\ 0 & 1 & 0 & 0\\ 1 & 0 & 0 & 0\\ 0 & 0 & 1 & 0\end{array}\right),~~ V_2 = \left(\begin{array}{cccc} 0  & 0 & 1 & 0\\ 0 & 0 & 0 & 1\\ 0 & 1 & 0 & 0\\ 1 & 0 & 0 & 0\end{array}\right),~~ V_3 = \left(\begin{array}{cccc} 1  & 0 & 0 & 0\\ 0 & 0 & 1 & 0\\ 0 & 0 & 0 & 1\\ 0 & 1 & 0 & 0\end{array}\right),~~ V_4 = \left(\begin{array}{cccc} 0  & 1 & 0 & 0\\ 1 & 0 & 0 & 0\\ 0 & 0 & 1 & 0\\ 0 & 0 & 0 & 1\end{array}\right) \nonumber \\
   & \Omega_1 = \left(\begin{array}{cccc} 1  & 0 & 0 & 0\\ 0 & -1 & 0 & 0\\ 0 & 0 & -1 & 0\\ 0 & 0 & 0 & 1\end{array}\right),~~ \Omega_2 = \left(\begin{array}{cccc} -1  & 0 & 0 & 0\\ 0 & -1 & 0 & 0\\ 0 & 0 & 1 & 0\\ 0 & 0 & 0 & 1\end{array}\right),~~ \Omega_3 = \left(\begin{array}{cccc} -1  & 0 & 0 & 0\\ 0 & 1 & 0 & 0\\ 0 & 0 & -1 & 0\\ 0 & 0 & 0 & 1\end{array}\right).
   \label{bobcorr4ics}
\end{align}

After the teleportation process, Alice ends up with either of the $\ket{0}_A,\ket{1}_A,\ket{2}_A,\ket{3}_A$ states and Charlie ends up with mass configuration in either of the $\ket{a},\ket{b},\ket{c},\ket{d}$ states. 

We will now describe the back-teleportation protocol from Bob to Alice. Let us consider an arbitrary state $\ket{\psi}_{B'B}=\sum_{i=0}^3\sum_{j=0}^{d-1}\alpha_{ij}\ket{i}_{B'} \ket{j}_B \in \mathbb{C}^4\otimes\mathbb{C}^d$ (where $\sum_{i,j}|\alpha_{ij}|^2=1$) of Bob's ancillary ququad and his $d$-dimensional subsystem.
For back teleportation, we assume that Alice resets her ququad to $\ket{0}_A$ state and Charlie resets the mass configuration in $\ket{a}=\frac{1}{2}(\ket{\mathcal{M}_{X_1  \rightarrow X_2 \rightarrow X_3 \rightarrow Y}} + \ket{\mathcal{M}_{X_1  \rightarrow X_2 \rightarrow Y \rightarrow X_3}} + \ket{\mathcal{M}_{X_1 \rightarrow Y \rightarrow X_2 \rightarrow X_3}} + \ket{\mathcal{M}_{Y \rightarrow X_1 \rightarrow X_2 \rightarrow X_3}})$ state. Now, Alice and Bob follow the same protocol for implementing respective local unitary operations on Alice's state and Bob's ancilla. The joint state of Alice's ququad system, Bob's systems and Charlie's mass configuration can be written as

\begin{align}
    \frac{1}{2}&\left(\ket{\mathcal{M}_{X_1  \rightarrow X_2 \rightarrow X_3 \rightarrow Y}} U_{X_3}U_{X_2}U_{X_1}\ket{0}_{A} (U_{Y|X_3}\otimes\mathbb{I})\ket{\psi}_{B'B} + \ket{\mathcal{M}_{X_1 \rightarrow X_2 \rightarrow Y \rightarrow X_3}} U_{X_3|Y}U_{X_2}U_{X_1}\ket{0}_{A} (U_{Y|X_2}\otimes\mathbb{I})\ket{\psi}_{B'B}\right. \nonumber \\
    &+ \left.\ket{\mathcal{M}_{X_1 \rightarrow Y \rightarrow X_2 \rightarrow X_3}} U_{X_3}U_{X_2|Y}U_{X_1}\ket{0}_{A} (U_{Y|X_1}\otimes\mathbb{I})\ket{\psi}_{B'B}+ \ket{\mathcal{M}_{Y \rightarrow X_1 \rightarrow X_2 \rightarrow X_3}}U_{X_3}U_{X_2}U_{X_1|Y}\ket{0}_{A} (U_{Y}\otimes\mathbb{I})\ket{\psi}_{B'B}\right) \nonumber \\
    = \frac{1}{4}&\left\{\ket{a}(\ket{0}_{A} \ket{\phi}_{B'B} + \ket{3}_{A}\ket{\lambda}_{B'B} + \ket{1}_{A} \ket{\mu}_{B'B} + \ket{2}_{A} \ket{\psi}_{B'B})\right. \nonumber\\
    &+\left.\ket{b}(\ket{0}_{A} \ket{\phi}_{B'B} - \ket{3}_{A}\ket{\lambda}_{B'B} - \ket{1}_{A} \ket{\mu}_{B'B} + \ket{2}_{A} \ket{\psi}_{B'B})\right. \nonumber \\
    &+ \left.\ket{c}(\ket{0}_{A} \ket{\phi}_{B'B} - \ket{3}_{A}\ket{\lambda}_{B'B} + \ket{1}_{A} \ket{\mu}_{B'B} - \ket{2}_{A} \ket{\psi}_{B'B})\right.\nonumber\\
    &+\left. \ket{d}(\ket{0}_{A} \ket{\phi}_{B'B} + \ket{3}_{A}\ket{\lambda}_{B'B} - \ket{1}_{A} \ket{\mu}_{B'B} - \ket{2}_{A} \ket{\psi}_{B'B})\right\},
\end{align}

Thereafter, Charlie measures the mass state in $\lbrace \ket{a},\ket{b},\ket{c},\ket{d} \rbrace$ basis and classically communicate his outcome $\lbrace a,b,c,d\rbrace$ to Alice; and Bob measures his ancillary system in $\lbrace |0\rangle_{B'} ,|1\rangle_{B'} , \ket{2}_{B'}, \ket{3}_{B'} \rbrace$ basis and communicate his outcome $i \in \lbrace 0,1,2,3\rbrace$ classically to Alice. Alice, upon receiving messages from Bob and Charlie, will accordingly make a correction operation on her subsystem. Consequently, the $3 \times d$-dimensional shared state between Alice and Bob becomes $\sum_{i=0}^3\sum_{j=0}^{d-1}\alpha_{ij}\ket{i}_{A}\ket{j}_B$. These details are summarized in the Table \ref{tab6}. Alice's correction operations mentioned in Table \ref{tab6} can be determined from the following expressions, 
\begin{align}
    &W_1 = \left(\begin{array}{cccc} 0  & 0 & 1 & 0\\ 1 & 0 & 0 & 0\\ 0 & 1 & 0 & 0\\ 0 & 0 & 0 & 1\end{array}\right),~~ W_2 = \left(\begin{array}{cccc} 0  & 0 & 0 & 1\\ 0 & 0 & 1 & 0\\ 1 & 0 & 0 & 0\\ 0 & 1 & 0 & 0\end{array}\right),~~ W_3 = \left(\begin{array}{cccc} 0  & 1 & 0 & 0\\ 0 & 0 & 0 & 1\\ 0 & 0 & 1 & 0\\ 1 & 0 & 0 & 0\end{array}\right),~~ W_4 = \left(\begin{array}{cccc} 1  & 0 & 0 & 0\\ 0 & 1 & 0 & 0\\ 0 & 0 & 0 & 1\\ 0 & 0 & 1 & 0\end{array}\right) \nonumber \\
   & \Omega_1 = \left(\begin{array}{cccc} 1  & 0 & 0 & 0\\ 0 & -1 & 0 & 0\\ 0 & 0 & -1 & 0\\ 0 & 0 & 0 & 1\end{array}\right),~~ \Omega_2 = \left(\begin{array}{cccc} -1  & 0 & 0 & 0\\ 0 & -1 & 0 & 0\\ 0 & 0 & 1 & 0\\ 0 & 0 & 0 & 1\end{array}\right),~~ \Omega_3 = \left(\begin{array}{cccc} -1  & 0 & 0 & 0\\ 0 & 1 & 0 & 0\\ 0 & 0 & -1 & 0\\ 0 & 0 & 0 & 1\end{array}\right).
   \label{Alicecorr4ics}
\end{align}

\end{appendix}

\end{document}